\documentclass[12pt]{article}

\usepackage[nottoc]{tocbibind}

\usepackage[colorlinks =true,linkcolor=blue,urlcolor=blue,citecolor=blue]{hyperref}

\usepackage{subcaption}
\usepackage{url}
\usepackage[utf8]{inputenc}
\usepackage{tabulary}
\usepackage{geometry}
\usepackage{fancyhdr}
\usepackage{setspace}
\usepackage{authblk}
\usepackage{natbib}
\usepackage{float}
\usepackage{multirow}
\usepackage{pdflscape}
\usepackage{graphicx}
\usepackage{adjustbox}
\usepackage{caption}
\usepackage{booktabs}
\usepackage{amsmath}
\usepackage{amsfonts}
\usepackage{amssymb}
\usepackage{amsxtra}
\usepackage{amsthm}
\usepackage{mathrsfs}
\usepackage{bbm}
\usepackage{enumerate}
\usepackage{enumitem}
\usepackage{xcolor}
\usepackage{lscape}

\newtheorem{lemma}{Lemma}
\theoremstyle{definition}
\newtheorem{assumption}{Assumption}
\theoremstyle{definition}

\theoremstyle{definition}
\newtheorem{definition}{Definition}[section]
\newtheorem*{remark}{Remark}
\theoremstyle{definition}

\newtheorem{prop}{Proposition}
\newtheorem{claim}{Claim}

\newcommand{\bitemsize}{\begin{itemize}}
\newcommand{\eitemsize}{\end{itemize}}
\newcommand{\be}{\begin{equation*}\begin{aligned}}  
\newcommand{\ee}{\end{aligned}\end{equation*}}
\newcommand{\benum}{\begin{enumerate}}
\newcommand{\eenum}{\end{enumerate}}
\newcommand{\bmat}{\begin{bmatrix}}
\newcommand{\emat}{\end{bmatrix}}

\newcommand{\E}{\mathrm{E}}

\title{Polarization by Design: How Elites Could Shape Mass Preferences as AI Reduces Persuasion Costs
}
\author{Nadav Kunievsky\footnote{Knowledge Lab, University of Chicago. 
I thank James A. Evans and Natalie Goldshtein for helpful discussions. 
This paper has greatly benefited from a review generated by \href{https://www.refine.ink/}{Refine.ink}. All remaining errors are my own.}}\date{\today}
\begin{document}

\maketitle
\begin{abstract}
In democracies, major policy decisions typically require some form of majority or consensus, so elites must secure mass support to govern. Historically, elites could shape support only through limited instruments like schooling and mass media; advances in AI-driven persuasion sharply reduce the cost and increase the precision of shaping public opinion, making the distribution of preferences itself an object of deliberate design. We develop a dynamic model in which elites choose how much to reshape the distribution of policy preferences, subject to persuasion costs and a majority rule constraint. With a single elite, any optimal intervention tends to push society toward more polarized opinion profiles—a ``polarization pull''—and improvements in persuasion technology accelerate this drift. When two opposed elites alternate in power, the same technology also creates incentives to park society in ``semi-lock'' regions where opinions are more cohesive and harder for a rival to overturn, so advances in persuasion can either heighten or dampen polarization depending on the environment. Taken together, cheaper persuasion technologies recast polarization as a strategic instrument of governance rather than a purely emergent social byproduct, with important implications for democratic stability as AI capabilities advance.
\end{abstract}

\section{Introduction}\label{sec:intro}

Democratic policymaking is mediated by public support. Constitutions, electoral rules, and informal norms all embed some version of a majority or consensus constraint: policy can be implemented only if enough citizens, or their representatives, are willing to endorse it. For elites who care about which policies are chosen, this makes the distribution of public preferences a central strategic object. To govern, they must either accept the beliefs they inherit or invest resources in reshaping them.

Historically, the tools available to shape mass preferences have been blunt, dull, and slow. States and parties have relied on school curricula, public broadcasting, subsidized media, patronage networks, and propaganda campaigns to tilt opinion in their favor.\footnote{See, among many others, \citet{DellaVigna2007FoxNews}, \citet{Adena2015Radio}, \citet{Cantoni2017Curriculum}, and \citet{dellavigna2010persuasion}.} These instruments operate at coarse levels of targeting, require long lead times, and are costly to adapt. As a result, even powerful governments have typically faced high and relatively rigid costs of persuasion.

The rapid diffusion of modern AI is changing this constraint. Generative models, agentic systems, and platform-level tools make it possible to generate, test, and personalize persuasive content at scale, in real time, and at very low marginal cost.\footnote{For recent evidence on LLM-based persuasion, see \citet{Salvi2025GPT4Persuasiveness}, \citet{Schoenegger2025LLMPersuasion}, \citet{tappin2023quantifying}, \citet{argyle2025testing}, \citet{bai2025llm}, \citet{simchon2024persuasive}, and \citet{hackenburg2024evaluating}.} A “nation of geniuses on the cloud’’ (\citet{korinek2024economics}) can now draft tailored messages, probe millions of individuals, and adapt narratives as feedback arrives. As these technologies lower and reshape persuasion costs, public preferences cease to be a fixed constraint and become something closer to a choice variable for elites. This raises a natural question: if elites can cheaply engineer the distribution of opinions, \emph{what distributions will they find optimal?}

In this paper, we develop a dynamic model in which, in each period, an elite faces a binary policy choice, must secure a majority to enact its preferred policy, and can, at some cost, shift public support. The model’s key mechanism is the uncertainty elites face about future ideal policies, which may change over time. Because public preferences are sticky and costly to alter, shaping opinion today influences an elite’s ability to enact its preferred policy in the future.

Our analysis focuses on the degree of polarization that would arise if elites could shape public preferences. We define polarization as the distance from unanimous agreement: in a binary setting, society is maximally polarized when support is evenly split and minimally polarized when near consensus. The central object of choice is not the policy itself but the distribution of mass preferences that governs future policy contests.

Our central departure from the existing polarization literature is conceptual. A large body of work studies polarization as the outcome of deeper forces such as changes in income distributions, identity cleavages, media markets, party strategies, or social networks—and then analyzes its consequences for turnout, policy, and welfare.\footnote{See, for example, \citet{DiMaggioEvansBryson1996}, \citet{EstebanRay1994}, \citet{DuclosEstebanRay2004},
\cite{McCartyPooleRosenthal2016}, \cite{IyengarWestwood2015}, \cite{Mason2018}, \cite{Shayo2009}, \cite{DellaVigna2007FoxNews}, \cite{MartinYurukoglu2017}, \cite{Lee2016}, \cite{Binder1999}, \cite{Azzimonti2011}, 
and the literature on affective polarization surveyed by \citet{iyengar2019origins}.} We instead treat polarization as the result of a \emph{choice}: a policy instrument in the hands of elites who face majority or consensus constraints and have the technology to reshape public opinion. In our framework, elites endogenously decide whether to maintain a cohesive society or to manufacture a more divided one, given their expectations about future states of the world and the technology of persuasion.

We begin by analyzing the case of a single, forward-looking elite. In each period, the state of the world determines which of two policies is ex post desirable. The elite observes this state and can adjust public support before the majority rule determines policy. With costly persuasion technology, we show a robust polarization pull: whenever the elite chooses to influence public opinion, society moves weakly toward maximal polarization. Intuitively, a distribution of opinions clustered near the majority threshold provides insurance against future uncertainty, allowing the elite to respond quickly and cheaply to changing priorities. If the state flips in the next period, the cost of shifting public support across the threshold is minimal. Polarization is therefore not intrinsically valuable to the elite but strategically attractive because a divided society is easier and cheaper to steer in the face of shocks. As persuasion costs decline—reflecting, for example, increasingly capable AI—the model predicts faster convergence toward highly polarized opinion profiles.

We then introduce a second elite with diametrically opposed preferences and alternating control. When today’s leader anticipates that tomorrow’s rival will also be able to reshape public opinion, the problem becomes more subtle. We show that competition introduces a new force that can work \emph{against} polarization. On the one hand, both elites still value a society near the majority threshold because it is cheap to move. On the other hand, a highly polarized society is easy for the rival to recapture when power switches. This creates incentives to push opinions into “semi-lock’’ regions that are costly for the opponent to overturn, and can generate a pull toward more cohesive, less polarized society. Whether  persuasion technology amplifies or dampens polarization in the two-elite case depends on the relative pull of the semi-lock regions and polarizations, compare to the status quo.

Our analysis relates to three strands of literature. First, we contribute to the extensive body of work on polarization and its political consequences (e.g., \cite{DiMaggioEvansBryson1996}; \citealp{EstebanRay1994}; \citealp{DuclosEstebanRay2004}) by modeling polarization as the optimal choice of a forward-looking decision-maker facing a consensus constraint, rather than as a reduced-form outcome of other processes. Second, we connect to the literature on the economics of persuasion and media, where actors invest in information or framing to influence beliefs and policy outcomes \citep[e.g.,][]{DellaVigna2007FoxNews,Adena2015Radio,Cantoni2017Curriculum,dellavigna2010persuasion}. We differ by focusing not on whether these persuasive efforts succeed, but on what the long-term objectives of such actors might be, and how advances in persuasion technology shape their behavior. Third, we contribute to the emerging literature on transformative AI (see the review by \citealp{korinek2024economics}), which examines how advanced AI may reallocate power, reshape information flows, and alter political equilibria. Our model offers a tractable mechanism linking improvements in persuasion technology, driven by AI (\cite{Salvi2025GPT4Persuasiveness}, \cite{Schoenegger2025LLMPersuasion},
\cite{tappin2023quantifying},
\cite{argyle2025testing},
\cite{bai2025llm},
\cite{simchon2024persuasive},
\cite{hackenburg2024evaluating}), to the long-run distribution of opinions and the dynamics of polarization.

The remainder of the paper is organized as follows. Section~\ref{sec:singleElite} introduces the single-elite framework, characterizes optimal policies in a two-period benchmark, and derives the polarization pull and its comparative statics in the infinite-horizon setting. Section~\ref{sec:twoElites} extends the model to two competing elites, analyzes a two-period Stackelberg game, and studies Markov-perfect equilibria numerically, highlighting the interaction between polarization and lock-in incentives. Section~\ref{sec:conclusions} concludes.

\section{Single Elite}\label{sec:singleElite}

Time is discrete, \( t \in \{1, \ldots, T\} \). In each period, the ruling elite can influence the distribution of public opinion. The elite’s objective depends on the current state of the world, denoted \( s_t \in \{0,1\} \). The state is drawn independently across periods, with \(\Pr(s_t = 1) = \pi \in (0,1)\). After observing the state of the world, the elite selects its preferred policy. We assume that the elite’s ideal policy matches the realized state.

To implement a policy \( y_t \in \{0,1\} \), the elite requires sufficient public support. Let \( p_t \) denote the share of the public that supports policy \(1\) in period \(t\). Define the implemented policy as
\[
y^s_t(p_t) = \begin{cases}
\mathbbm{1} & \text{ if } p_t > \tfrac{1}{2} \\
s & \text{ if } p_t = \tfrac{1}{2} \\
0 & \text{Otherwise}
\end{cases}.
\]
where we assume that if $p_t = \frac{1}{2}$, the elite can choose which policy to implement. This reflects the idea that they only need to sway a minuscule share of the public to enact their preferred policy. Throughout, we slightly abuse notation by writing $y_t(p_t)$ instead of $y_t^{s}(p_t)$, understanding that the ruling elite can implement its preferred policy when $p_t = \tfrac{1}{2}$.
 
The elite’s payoff from policy \( y_t \), given state \( s_t \), is
\[
u(y_t, s_t) = H \cdot \mathbbm{1}\{y_t(p_t) = s_t\},
\]
where \( \mathbbm{1}{\cdot} \) is the indicator function and \( H \) is a scalar representing the benefit to the elite from implementing its preferred policy. The elite has also a technological capacity to shape public opinion. Let \( c(p' - p) \) denote the cost of changing the share of the public that supports policy \(1\) from \( p \) to \( p' \). We make the following assumption on the cost function. 

\begin{assumption}\label{ass:cost_assumption}
The cost function \( c(x) \) is strictly increasing in \( |x| \), strictly convex, satisfies \( c(0) = 0 \), and is symmetric such that for any \( p, p' \in [0,1] \), $c(p' - p) = c(p - p')$.
\end{assumption}

\begin{remark}
Our model is agnostic about the exact channels through which governments shape public beliefs. In practice, elites influence opinion via (i) information provision and control (e.g., \citet{kamenica2011bayesian,PratStromberg2013,EnikolopovPetrovaZhuravskaya2011,GentzkowShapiroSinkinson2011}), (ii) incentives that reward compliance or mobilize turnout (e.g., \citet{Stokes2005,DeLaO2013,ManacordaMiguelVigorito2011}), and (iii) direct preference formation through schooling and value-oriented interventions (e.g.,\citet{AlesinaFuchsSchundeln2007,Cantoni2017Curriculum}). As technology advances, this arsenal extends beyond education, subsidies to pivotal groups, coercion, and mass media (e.g., \cite{Adena2015Radio,DellaVigna2007FoxNews,dellavigna2010persuasion}) to include algorithmically amplified, adaptive messaging—“transformative AI”—that personalizes, scales, and optimizes persuasion (e.g \cite{Salvi2025GPT4Persuasiveness}, \cite{Schoenegger2025LLMPersuasion},
\cite{tappin2023quantifying},
\cite{argyle2025testing},
\cite{bai2025llm},
\cite{simchon2024persuasive},
\cite{hackenburg2024evaluating}). These AI-driven tools enable elites to spread targeted information or misinformation more efficiently, leveraging the structure of influence itself rather than direct persuasion. The difficulty of shifting the beliefs of a large population is captured by our assumption of a convex cost function, implying that changing the minds of a small group is substantially easier than transforming the views of society at large.
\end{remark}

\begin{remark}
Our model leaves the elite’s objective unspecified. It can accommodate a wide spectrum of motivations—from public-regarding to purely self-interested—and hybrids in between (e.g., paternalistic, partisan, ideological, or reputation-driven goals). To fix ideas, consider two polar illustrations. On one end, a benevolent elite aims to adopt the policy that best serves overall welfare but can do so only once sufficient consensus forms. Because individuals revise their views gradually as conditions and information evolve, popular support may lag the welfare-improving choice; the elite’s role is to inform and persuade so that consensus eventually endorses it. On the other end, a rent-seeking elite advances policies that advantage its own members even at the expense of aggregate welfare; our model is silent about which one of these cases take place, it simply assumes that the elites have some ideal policy in mind. 
\end{remark}

Throughout the analysis, we operationalize polarization in the public as the distance from complete consensus—that is, from either \(p = 0\) or \(p = 1\). In our binary setting, a convenient “distance-from-consensus” index is $\frac{1}{2}-\bigl|p-\tfrac12\bigr|$, which is maximal when opinions are evenly split, \(p=\tfrac12\), and declines symmetrically as the distribution of opinions concentrates near one of the extremes. This formulation corresponds to the standard treatment of polarization as dispersion or disagreement in collective beliefs: a society is minimally polarized when a single position is widely held, and maximally polarized when the population is divided into two equally sized and mutually opposed groups (\cite{DiMaggioEvansBryson1996}).

In our binary setting this notion of polarization is closely related to variance-based measures of polarization commonly used. In our case, the variance in public opinion is given by
\[
\operatorname{Var}(\text{opinion}) = p(1-p),
\]
which attains its maximum at \(p = \tfrac{1}{2}\) and declines as \(p\) moves toward either 0 or 1. Note that these measures are all equivalent up to strictly monotone transformation: \(p(1-p)=\tfrac14-(p-\tfrac12)^2\) is strictly increasing in \(d(p)\) and strictly decreasing in \(|p-\tfrac12|\). Accordingly, since our results are ordinal (comparative statics), nothing depends on the particular formula, and we use whichever representation is algebraically most convenient.

\begin{remark}
It is important to note, however, that the broader literature has advanced several alternative conceptions of polarization. Some contributions emphasize identification with distinct poles and derive indices that explicitly reward both intra-group homogeneity and inter-group separation (\cite{EstebanRay1994,DuclosEstebanRay2004}). Others stress the structure of disagreement across social or partisan groups. For instance, measuring the distance between partisan means rather than overall variance (\cite{Dalton2008}). A more recent strand studies \emph{affective} polarization, namely the intensification of negative evaluations of out-groups, which can rise even when the distribution of policy views remains stable (\cite{iyengar2019origins}). Our choice to work with the variance-like measure should therefore be read as an analytically convenient specification that aligns with our binary framework, rather than as a claim that it exhausts all the concept of social polarization.   
\end{remark}

\subsection{Simple Two-Period Model}
We begin with the simplest case, where there are two periods, $T = 2$. The elite’s optimization problem can be expressed as a pair of Bellman equations, written as functions of the initial public opinion and the current state:
\begin{align*}
V_{1,s}(p) &:= \max_{p'\in[0,1]} \, u(y(p'), s) - c(p'-p)
\;+\; \beta\,\left[\pi V_{2,1}(p') + (1-\pi)V_{2,0}(p')\right] \\ 
V_{2,s}(p) &:= \max_{p'\in[0,1]} \, u(y(p'), s) - c(p'-p) .
\end{align*}

We solve for the optimal policy using backward induction.
Let $\Delta = c^{-1}(H)$ denote the positive solution to $c(\Delta) = H$.
The following claim then characterizes the elite’s choice in the second period.  
\begin{claim}[Period-2 policy and value]\label{claim:oneElite_period2}
Define the cutoffs
\[
p_0^* \;:=\; \frac{1}{2} - \Delta, 
\qquad
p_1^* \;:=\; \frac{1}{2} + \Delta.
\]
Then the period-2 optimal public opinion distribution, $p_s^*(p,s)$, is given by:
\[
p_2^*(p,1)=
\begin{cases}
p, & p\ge \tfrac12 \ \text{or}\ p\le p_0^*,\\
\tfrac12, & p_0^*<p<\tfrac12,
\end{cases}
\qquad
p_2^*(p,0)=
\begin{cases}
p, & p\le \tfrac12 \ \text{or}\ p\ge p_1^*,\\
\tfrac12, & \tfrac12<p<p_1^*.
\end{cases}
\]
The associated period-2 values are
\[
V_{2,1}(p)=
\begin{cases}
H, & p\ge \tfrac12,\\
H-c(\tfrac12-p), & p\in[p_0^*,\,\tfrac12],\\
0, & p<p_0^*,
\end{cases}
\qquad
V_{2,0}(p)=
\begin{cases}
H, & p\le \tfrac12,\\
H-c(p-\tfrac12), & p\in[\tfrac12,\,p_1^*],\\
0, & p>p_1^*.
\end{cases}
\]
\end{claim}
The proof is in Appendix \ref{app:proofClaim1}.
Claim \ref{claim:oneElite_period2} describes the elite’s optimal choice of public opinion in the second period, when there are no future payoffs or uncertainty. The claim implies that the elite leaves opinion unchanged if it already favors its preferred policy or if shifting it would be too costly. Otherwise, when influence costs are moderate, the elite adjusts public sentiment just enough to reach $\tfrac{1}{2}$, the minimal threshold for implementing its desired policy.

We can now describe the optimal public opinion that the elite would choose in the first period, when facing future uncertainty on the second period optimal policy. First, It will be convenient to partition \([0,1]\) into the four regions
\[
A=[0,p_0^*],\qquad B=[p_0^*,\tfrac12],\qquad C=[\tfrac12,p_1^*],\qquad D=[p_1^*,1].
\]
Using claim \ref{claim:oneElite_period2}, we can calculate the expected continuation value in period 2 as a function of the chosen public opinion share, \(p'\):
\begin{equation}
\label{eq:EV2_piecewise}
\E_{s'}[V_{2,s'}(p')] =
\begin{cases}
H(1-\pi), & p'\in A,\\[2pt]
H-\pi\,c(\tfrac12-p'), & p'\in B,\\[2pt]
H-(1-\pi)\,c(p'-\tfrac12), & p'\in C,\\[2pt]
H\pi, & p'\in D.
\end{cases}
\end{equation}
We further define 
\begin{align*}
p_{B,\max}(p) &\in \arg\min_{q\in [p_0^*,\,\tfrac12]}\Big\{\, c(q-p)+\beta\,\pi\,c(\tfrac12-q)\,\Big\},\\
p_{C,\max}(p) &\in \arg\min_{q\in [\tfrac12,\,p_1^*]}\Big\{\, c(q-p)+\beta\,(1-\pi)\,c(q-\tfrac12)\,\Big\}.
\end{align*}
which we know exist, by assumption \ref{ass:cost_assumption}. The following claim describes the optimal decision in the first period.

\begin{claim}[Single Elite optimal Strategy]\label{claim:single_elite_two_period}
Fix inherited support \(p\in[0,1]\) and state \(s\in\{0,1\}\) in period 1. 
Then:
\begin{enumerate}
\item[(i)] (\emph{Candidate set}) An optimal period-1 choice \(p'_1{}^*(p,s)\) belongs to
\[
\big\{\, p,\; p_{B,\max}(p),\; p_{C,\max}(p),\; \tfrac12 \,\big\}.
\]
\item[(ii)] (\emph{Polarization pull}) Every optimal \(p'\) satisfies
\[
\big|\,p'-\tfrac12\,\big|\;\le\;\big|\,p-\tfrac12\,\big|.
\]
Equivalently, the elite either leaves support unchanged or moves it weakly toward higher polarization, \(\tfrac12\), but never strictly away from it.
\end{enumerate}

\end{claim}
The proof is in Appendix \ref{app:proofClaim2}. Claim \ref{claim:single_elite_two_period} states that, given the initial distribution \(p_0\in[0,1]\) and realized state \(s\in \{0,1\}\) in period $1$, the elite’s optimal target \(p_1\) lies in the four-point candidate set \(\{p_0, p_{B,\max}(p_0), p_{C,\max}(p_0), \tfrac12\}\). These correspond to three move types: (i) inaction, \(p_1 = p_0\); (ii) a modest shift toward the polarization peak at \(\tfrac12\) via the interior maximizers \(p_{B,\max}(p_0)\) or \(p_{C,\max}(p_0)\); or (iii) a discrete jump to the middle, \(p_1 = \tfrac12\), which yields maximal polarization. Moreover, any optimal adjustment weakly reduces \(|p_1 - \tfrac12|\): the elite never moves opinion strictly away from the median and always, weakly, toward maximal polarization.

The intuition is straightforward. Under majority rule and uncertainty about next period’s optimal policy, the elite has an incentive to polarize public opinion. A more polarized distribution reduces the swing margin required if the preferred policy changes in the next period, making future policy pivots cheaper. Consequently, a more fractionalized society improves the elite's ability to respond to uncertainty.

Figure \ref{fig:twoPeriodDemonstration}(a)–(b) illustrates the first-period trade-off. It decomposes the elite’s problem into the benefit of choosing \(p_1\) and the cost, given \(p_0\), under each state \(s \in \{0,1\}\). In both panels, the solid blue curve shows the period-1 benefit of \(p_1\) (current payoff plus discounted expected continuation value), and the dashed red curve shows the convex adjustment cost \(c(p_1 - p_0)\). The optimal \(p_1\) maximizes the vertical gap between benefit and cost. As the figure makes clear, absent adjustment costs the present value of benefits is maximized at \(p_1 = \tfrac12\): maximal polarization guarantees implementation of the elite’s preferred policy in the current period and, if needed, in the next.

The figure highlights two regimes. (i) When \(p_1\) lies far from \(\tfrac12\) (regions $A$ or $D$), inaction in the second period \((p_2 = p_1)\) yields locally flat benefits in the current period. (ii) For intermediate \(p_1\) (regions $B$ or $C$), the value is increasing as \(p_1\) approaches \(\tfrac12\), because the elite will need to move opinion less in the next period, if needed, to implement their preferred policy. Where the elite ultimately positions society depends on adjustment costs and initial beliefs: when beliefs are highly consensual (close to 0 or 1), the cost of creating polarization may outweigh its insurance value, whereas when society is already somewhat polarized, the elite finds it attractive to further increase polarization to hedge against future policy shocks.

Even when the elite and the public are initially aligned on the optimal policy, forward-looking considerations can justify additional polarization. Tilting mass toward $\tfrac12$ increases the continuation value by keeping future policy changes cheap when shocks occur. In this sense, maintaining a polarized society is a way to preserve flexibility in the face of uncertainty.

\begin{figure}[H]
    \centering
    \includegraphics[width=\linewidth]{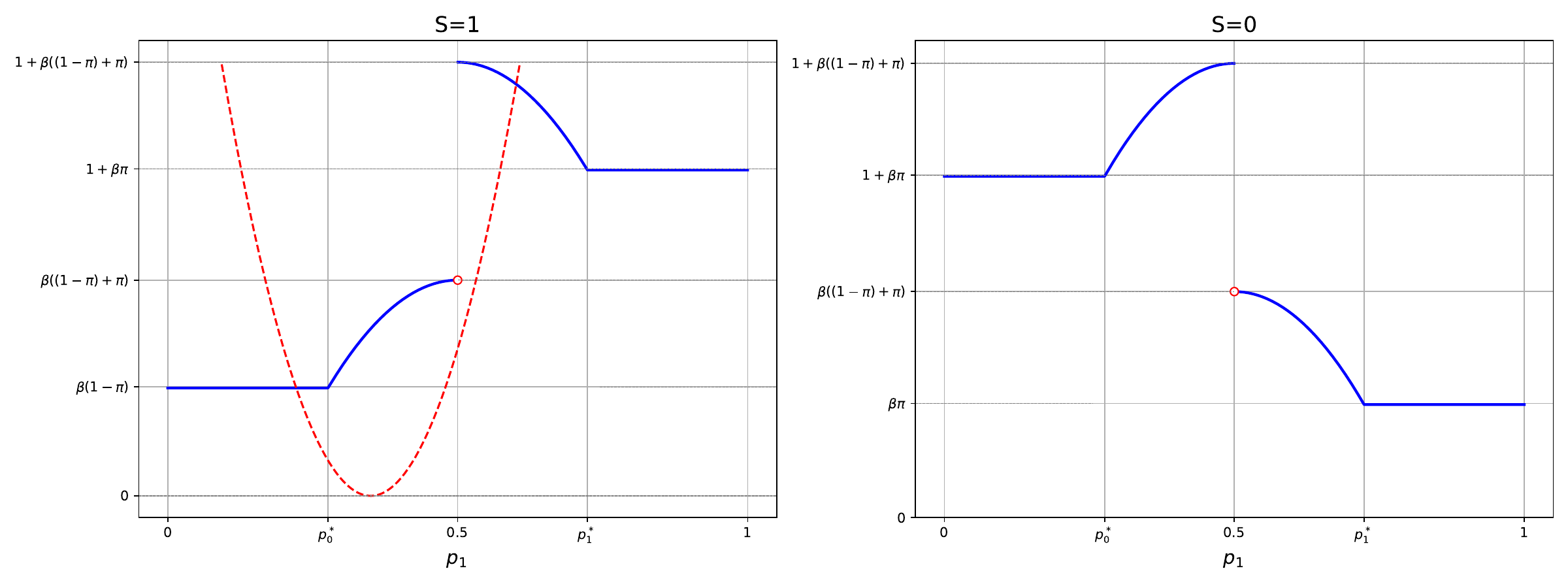}
    \caption{Value Function Benefits and Costs}
    \label{fig:twoPeriodDemonstration}
\end{figure}

\subsection{One Elite with Infinite Horizon}\label{subsec:one-elite-infinite}
We extend our simple two-period game to a long-horizon setting to demonstrate the generalizability of our mechanism and to show that the effect is not driven by terminal values. We maintain the same assumptions as before, letting $T \to \infty$. Let $V_s(p)$ denote the value function when the inherited support is p and the current state is s. The Bellman equation is:
\begin{equation}\label{eq:oneEliteInfiniteHorizon}
V_s(p)\ =\ \max_{p'\in[0,1]}\Big\{u(y(p'),s)-c(p'-p)\ +\ \beta\big[\pi\,V_1(p')+(1-\pi)\,V_0(p')\big]\Big\}.
\end{equation}
Let \( \sigma(s,p) \) be the policy function for the elite given state $s$ and $p$. The following proposition shows that the elite would weakly converge over time to $\frac{1}{2}$. 

\begin{prop}[Single Elite with Infinite Horizon]\label{claim:SingleInfiniteHorizon}
The value function \( V_s(p) \) and the optimal policy \( \sigma(s,p) \) satisfy:
\begin{enumerate}
    \item \textbf{Peak at the median.} For each \(s\in\{0,1\}\), \(V_s(\cdot)\) is weakly increasing on \( [0,\tfrac12] \), weakly decreasing on \( [\tfrac12,1] \), and attains a (not necessarily unique) maximum at \( p=\tfrac12 \).
    \item \textbf{Polarization pull and monotone policy.} For every \(s\) and \(p\),
    \[
    \min\{p,\tfrac12\}\ \le\ \sigma(s,p)\ \le\ \max\{p,\tfrac12\},
    \qquad\text{equivalently}\qquad
    |\sigma(s,p)-\tfrac12|\ \le\ |p-\tfrac12|.
    \]
    Moreover, on each side of the median the optimal policy is monotone in the inherited support:
    \[
      \text{if }0\le p<\bar p\le \tfrac12,\quad \sigma(s,p)\ \le\ \sigma(s,\bar p)\ \le\ \tfrac12,
      \qquad
      \text{and symmetrically for }\tfrac12\le \bar p<p\le 1.
    \]
    In particular, \(p=\tfrac12\) is absorbing: \(\sigma(s,\tfrac12)=\tfrac12\).
\end{enumerate}
\end{prop}

The proof is in Appendix \ref{app:proofProp1}. The proposition shows that the intuition from the two-period case extends to the general dynamic setting: Social polarization continues to exert a strong pull. Part (i) establishes that the value function is maximized when public opinion is most polarized, reflecting a force toward societies that are easy to shift in response to changes in the state of the world. Part (ii) shows that elites optimally move public opinion weakly toward \(p = \tfrac12\), thereby moving public opinion toward the point of maximal polarization.

The next claim demonstrates how the speed in which the elite increase polarization increase as the cost of shaping the public opinion goes down. We start with the following definition  
\begin{definition}[Cost dominance]\label{ass:cost-dominance}
Let $c,\tilde c:\mathbb{R}\to\mathbb{R}_+$ satisfy Assumption~\ref{ass:cost_assumption}.
We say that $\tilde c$ \emph{cost-dominates} $c$ if for all $0\le x_1<x_2$,
\[
\big[\tilde c(x_2)-\tilde c(x_1)\big]\ >\ \big[c(x_2)-c(x_1)\big].
\]
\end{definition}

This definition captures our notion of higher cost function, in the sense that $\tilde c$ has uniformly larger marginal increments away from the status quo (on the relevant side), i.e., it is ``more convex'' in the sense of increasing differences. The following claim shows that more convex cost implies a slower convergence.

\begin{claim}[One-step move shrinks under costlier technology]\label{lem:shrink-move}
Fix a state $s\in\{0,1\}$ and an inherited support $p\in[0,1]$. Define the one-step objectives
\[
\Psi_c(p';s,p):=R(s,p')-c(p'-p)+\beta \left[\pi V_1(p') + (1-\pi)V_0(p')\right],
\]
and equiviantly for $\Psi_{\tilde{c}}$. Let $p^*_c\in\arg\max_{p'\in[0,1]}\Psi_c(p';s,p)$ and $p^*_{\tilde c}\in\arg\max_{p'\in[0,1]}\Psi_{\tilde c}(p';s,p)$.
If $\tilde c$ cost-dominates $c$, then
\[
\text{if }p\le\tfrac12,\quad p^*_{\tilde c}\ \le\ p^*_c;\qquad
\text{if }p\ge\tfrac12,\quad p^*_{\tilde c}\ \ge\ p^*_c.
\]
\end{claim}

The proof is in Appendix \ref{app:proofClaim3}. The claim shows how declining persuasion costs accelerate elites’ movement toward polarization. As AI and large language models enable more precise and large-scale targeting, at lower costs, a single elite increasingly prefers a polarized environment in which public consensus is easier to shape. Incentives for polarization have always been present, but historically high costs of persuasion—due to inefficient technology—slowed the process. Claim \ref{lem:shrink-move} shows that as these costs fall, elites move society toward maximal polarization more rapidly.

To understand how convergence to high polarization depends on model parameters, we solve for the value function under a quadratic adjustment cost \(k(p' - p)^2\). Figure \ref{fig:oneEliteInfHorizon} plots the resulting policy functions for various parameterizations and for both states of the world. Across specifications, the policy converges toward the maximally polarized benchmark. In particular, when the initial majority preference (p) disagrees with the elite’s preferred policy, the elite pushes more aggressively toward polarization, as seen in the steep jumps toward \(p'=\tfrac12\) for higher values of \(p\) in state \(s=0\) and for lower values of \(p\) in state \(s=1\).

Figure \ref{fig:oneEliteInfHorizon_valueFunction} in the Appendix reports the corresponding value functions. The figure shows a pronounced inverted-U shape, with the value maximized when society is as polarized as possible. The figure also illustrates that the value function is generally neither concave nor differentiable. The non-differentiability reflects the majority-rule payoff structure, which creates sharp payoff jumps in response to marginal changes in public opinion around the majority threshold.

We now examine how changes in the perssuasion costs affects the elites optimal policy. Panels \(a\) and \(b\) of Figure \ref{fig:oneEliteInfHorizon} vary the persuasion cost parameter. The green line corresponds to high persuasion costs. In this case, the elite leaves most belief profiles essentially unchanged: adjustment is too expensive except when initial opinions are already divided, in which case the elite is willing to incur the high cost to push further toward maximal polarization and secure its preferred policy. As costs decline—capturing improved persuasion technologies enabled by AI—the orange and blue lines show that elites intervene over a much wider range of initial beliefs, shifting public opinion more quickly toward highly polarized configurations. These panels also highlight that when current public opinion is misaligned with the elite’s ideal policy, incentives to polarize are particularly strong: in the blue case, whenever the majority initially opposes the elite’s preferred policy, the elite chooses a rapid move toward the polarized benchmark.

Panels (c) and (d) vary a different dimension of AI-driven change: improved forecasting of future states. We capture this by changing \(\pi\), the probability of one state of the world. As \(\pi\) increases, the environment becomes more predictable, which alters the “insurance” role of polarization. When the likely state already prevails (e.g., we are in state 1 and \(\pi\) is high), reduced uncertainty weakens the incentive to adjust public beliefs: the elite can expect its preferred policy to remain optimal without incurring adjustment costs. By contrast, when society is temporarily in the *less likely* state, a higher \(\pi\) raises the expected likelihood that this state will soon be overturned. Even though the elite still chooses \(p' \le \tfrac12\) so that the current policy continues to prevail, it becomes more valuable to move opinions *closer* to the threshold \(p'=\tfrac12\) today, because doing so reduces the adjustment cost of crossing the threshold when the more likely state re-materializes and the long-run optimal policy flips. These transitional episodes—from a relatively unlikely state to a more stable one—therefore generate spikes in polarization as elites reposition public opinion in anticipation of the impending regime change.

Finally, Figures \ref{fig:singleEliteBetaH} and \ref{fig:singleEliteBetaH_valueFunction} in the Appendix examine how the discount factor \(\beta\) and the stakes of policy (through \(H\)) shape the single-elite policy function. In all cases, the core “polarization pull’’ from Proposition 1 remains; \(\beta\) and \(H\) only scale its strength. Panels (a) and (b) vary \(\beta\). When the discount factor is low \((\beta = 0.5)\), the elite is effectively more myopic: the policy function lies close to the 45-degree line away from the majority threshold, and large regions of \(p\) map into almost no adjustment. The elite pays persuasion costs mainly when current support lies on the “wrong’’ side of the threshold and flipping today’s policy is sufficiently valuable. As \(\beta\) rises, the policy curves bend more sharply toward \(p'=\tfrac12\): the inaction region shrinks, and the set of initial beliefs for which the elite chooses a discrete jump toward the median expands.

Panels (c) and (d) vary the per-period payoff \(H\). When \(H\) is small \((H = 0.5)\), the benefit–cost gap from shifting opinions is modest, so elites tolerate a wide range of inherited beliefs without intervening. As (H) increases, the policy functions move closer to the maximally polarized outcome \(p'=\tfrac12\). The cutoffs at which the elite is willing to pay for adjustment move outward, and more initial conditions trigger a substantial jump toward \(p' \approx \tfrac12\).

Taken together, figure \ref{fig:singleEliteBetaH} shows that both greater patience and higher stakes strengthen the polarization pull in the single-elite environment. When the elite expects to remain influential for a long time or faces very high returns to policy alignment, it invests more aggressively in shifting public opinion into highly polarized regions that minimize expected future adjustment costs.

\begin{figure}[H]
    \centering
    \includegraphics[scale=0.3]{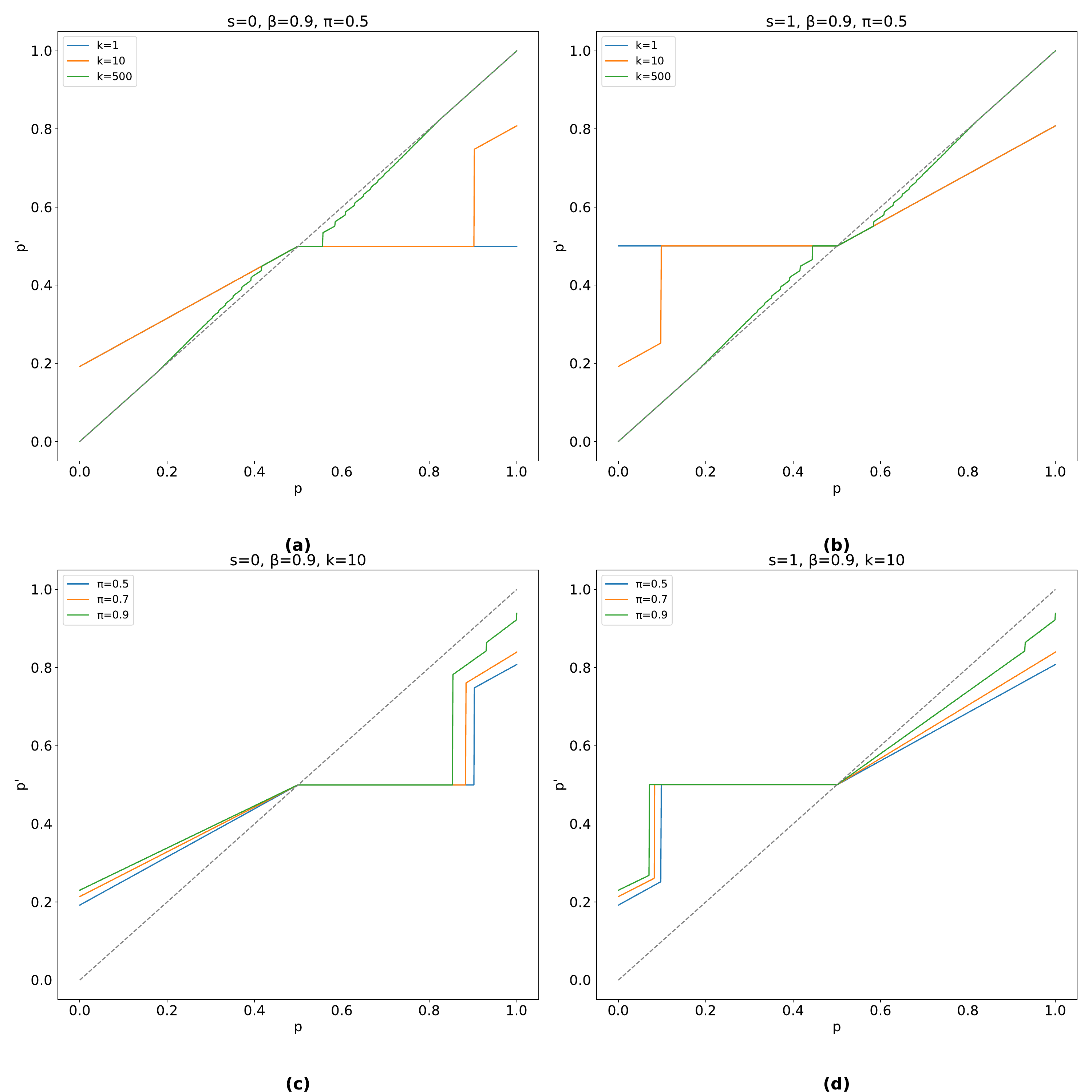}
    \caption{Single Elite - Policy Function}
    \label{fig:oneEliteInfHorizon}
\end{figure}

\section{Two Competing Elites}\label{sec:twoElites}

The previous results showed that a single elite governing through mass acceptance has strong incentives to increase polarization among the general public, as doing so provides insurance against future shifts in its ideal policy. We now turn to the case of two competing elites and examine how their interaction shapes the optimal preference distribution each elite chooses.

Time is discrete, $t = \{0, 1, 2, \ldots ,T \}$. At each date $t$, a state $s_t \in \{0,1\}$ is drawn independently and identically distributed with probability $\Pr(s_t = 1) = \pi \in (0,1)$, and is observed by the ruling elite (the “mover”). Two elites, $A$ and $B$, alternate in power deterministically: elite $A$ governs at odd periods, and elite $B$ governs at even periods.

The public’s preference state is $p_{t} \in [0,1]$. The ruling elite chooses the current period’s preference $p_{t+1} \in [0,1]$ at a cost $c(p_{t+1} - p_{t})$, where the cost function satisfies Assumption \ref{ass:cost_assumption}. Notice that $p_{t+1}$ represents the period $t$ public opinion that governs which policies can be implemented in that period, as well as the initial level of public support for policy 1 in period $t+1$. Both elites share the same cost function. The policy implemented in each period is $y = 1$ if $p_{t+1} > \tfrac{1}{2}$. When $p_{t+1} = \tfrac{1}{2}$, the mover selects $y \in \{0,1\}$. This setup reflects the idea that the ruling elite each period holds substantially greater influence over public opinion than the non-ruling opposition.

Elites $A$ and $B$ are assumed to have fundamentally opposing preferences: in every state of the world, their ideal policies diverge. Specifically, elite $A$ seeks to \emph{match} the state, while elite $B$ seeks to \emph{mismatch} it:
\[
u_A(y,s) = H \cdot \mathbbm{1}\{y = s\}, \qquad
u_B(y,s) = H \cdot \mathbbm{1}\{y = 1 - s\}, \qquad H > 0.
\]
Finally, both elites share the same discount factor, $\beta \in (0,1)$.

\subsection{Two-Period Stackelberg Equilibrium}
We begin by analyzing a simple two-period setting, $t \in \{1,2\}$. Elite $A$ moves in period $1$ (from an initial $p_0$) after observing $s_1$, while elite $B$ moves in period $2$ after observing $s_2$ and given the move of elite $A$. Let $\sigma_{A,s_1}$ and $\sigma_{B,s_2}$ denote the strategies of elites $A$ and $B$, respectively, under each state of the world.

Elite $A$’s objective, for a given state $s_1$, can be expressed as the following optimization problem:
\[
\max_{p_1\in [0,1]} \; \left\{ u_A\!\big(y(p_1),s_1\big)
	-	c(p_1 - p_0) \;+\; \beta\,\E_{s_2}\Big[\,u_A\big(y(\sigma_{B,s_2}(p_1)),s_2\big)\Big]\right\},
\]
while elite B solves, in the second period,
\[
\max_{p_2 \in [0,1]} \left\{ u_B\big(y(p_2),s_2\big) - c(p_2 - p_1) \right\}, \quad \text{for each } s_2.
\]
The Stackelberg equilibrium of this game is characterized by a pair of policy functions:
$\sigma_{A,s_1} : [0,1] \to [0,1]$ for elite $A$, and $\sigma_{B,s_2} : [0,1] \to [0,1]$ for elite $B$, that jointly solve the optimization problems of the two elites.

We again denote by $\Delta = c^{-1}(H)$ the positive solution to $c(\Delta) = H$. The following claim characterizes the equilibrium strategies in the Stackelberg equilibrium.

\begin{claim}[Two-Period Stackelberg equilibrium]\label{claim:2p_stackleberg}
Fix \(p_0\) and \(s_1\). Elite \(A\) chooses one of the four cases \(\{p_0,\,\tfrac12,\,\tfrac12\pm\Delta\}\) by comparing the value of these four scalars
\[
\begin{aligned}
&\text{Inaction:} && u_A\!\big(y(p_0),s_1\big) \;+\; \beta\cdot \Phi(p_0),\\
&\text{Median:} && H - c\big(p_0-\tfrac12\big) \;+\; \beta\cdot 0,\\
&\text{Semi-lock right:} && u_A(1,s_1) - c\!\big((\tfrac12+\Delta)-p_0\big) \;+\; \beta\cdot \pi H,\\
&\text{Semi-lock left:} && u_A(0,s_1) - c\!\big(p_0-(\tfrac12-\Delta)\big) \;+\; \beta\cdot (1-\pi)H,
\end{aligned}
\]
where
\[
\Phi(p_0)=
\begin{cases}
(1-\pi) H, & p_0 \le \tfrac{1}{2}-\Delta,\\
\pi H, & p_0 \ge \tfrac{1}{2}+\Delta,\\
0, & |p_0-\tfrac{1}{2}| < \Delta.
\end{cases}
\]
and Elite B have the following strategy
\[
\sigma_{B,s_2}(p_1)=
\begin{cases}
p_1, & \text{if } \mathbbm{1}\{p_1 \ge \tfrac12\}=1-s_2,\\[3pt]
\tfrac12, & \text{if } \mathbbm{1}\{p_1 \ge \tfrac12\}=s_2 \text{ and } |\tfrac12-p_1 |<  \Delta,\\[3pt]
p_1 & \text{if } \mathbbm{1}\{p_1 \ge \tfrac12\}=s_2 \text{ and } |\frac{1}{2} - p_1| \geq \Delta,
\end{cases}
\]

\end{claim}

The proof is in Appendix \ref{app:proofClaim4}. Claim \ref{claim:2p_stackleberg} characterizes the strategic logic of the two–periods game between elites $A$ and $B$. In the second period, elite $B$ either leaves public opinion at $p_1$ or nudges it just enough to flip the majority, provided the cost of doing so is not too high. Anticipating this, the first–period leader, elite $A$, chooses among four options. Elite $A$ can (i) stay put, (ii) jump directly to maximal polarization, knowing that elite $B$ may exploit this in the next period, or (iii–iv) move to one of two “semi–lock” positions that lie just outside the range in which elite $B$ can cheaply induce its ideal policy. These semi–lock locations are valuable precisely because they limit elite $B$’s ability to profitably deviate and implement policies that elite $A$ dislikes.

Rivalry between elites adds a strategic dimension: each seeks to shape public preferences to constrain the rival’s response. In the single-elite case, without strategic interaction, polarization is pushed to its extreme to insure against future uncertainty in the ideal policy. With two polarized elites, however, the ruling elite faces a countervailing incentive—to “lock in” public opinion and temper polarization—to prevent the rival from later advancing unfavorable policies.

To build intuition, Figure \ref{fig:benefit_cost_2_elites}, analogous to Figure \ref{fig:twoPeriodDemonstration}, decomposes elite $A$’s period 1 payoffs into benefits and costs from each $p_1$, given some initial $p_0$. The blue line shows the sum of the current payoff and the discounted continuation value from choosing different values of $p_1$, for each state ($s_1\in\{0,1\}$). The red dashed line shows the cost of choosing $p_1$ when $p_0=0.35$. As in Figure \ref{fig:twoPeriodDemonstration}, the net value of choosing $p_1$ is the vertical difference between the blue and red lines.

The plot shows four strategic regions: (i) to the left of \(\tfrac{1}{2}-\Delta\), (ii) a left neighborhood of \(\tfrac{1}{2}\), (iii) a right neighborhood of \(\tfrac{1}{2}\), and (iv) to the right of \(\tfrac{1}{2}+\Delta\). These regions are determined by elite $B$’s willingness to intervene and change public opinion. Far enough from \(\tfrac{1}{2}\), elite $B$ either already obtains its preferred policy or finds overturning the majority too costly; in these regions, elite $A$’s continuation value depends only on whether the optimal policy next period remains unchanged. In contrast, near \(\tfrac{1}{2}\), elite $B$ can profitably flip the majority and implement a policy that elite $A$ dislikes, so elite $A$’s continuation value falls to zero. This generates the downward steps at \(\tfrac{1}{2}\pm\Delta\).

Comparing Figures \ref{fig:benefit_cost_2_elites} and \ref{fig:twoPeriodDemonstration} shows that, with two elites, the maximum net benefit is no longer attained at maximal polarization. Instead, the lock–in region, where social cohesion is high, delivers the highest value. Taking elite $B$’s incentives into account thus pushes the current elite toward preserving cohesion rather than maximizing polarization. The choice of $p_1$ is determined by comparing the net benefits in each region to the cost of moving away from $p_0$. 

The cost function determines the size of the deterrence regions. Higher costs enlarge the range in which elite $A$ prefers not to move, while lower costs make it attractive for both elites to intervene more aggressively to reshape public opinion and implement their preferred policies. Moreover, lower costs shrink the “lock-in’’ region at the extremes, implying that elite $A$, when optimal given the cost, would choose a lower level of public polarization.

\begin{figure}[H]
    \centering
    \includegraphics[scale=0.5]{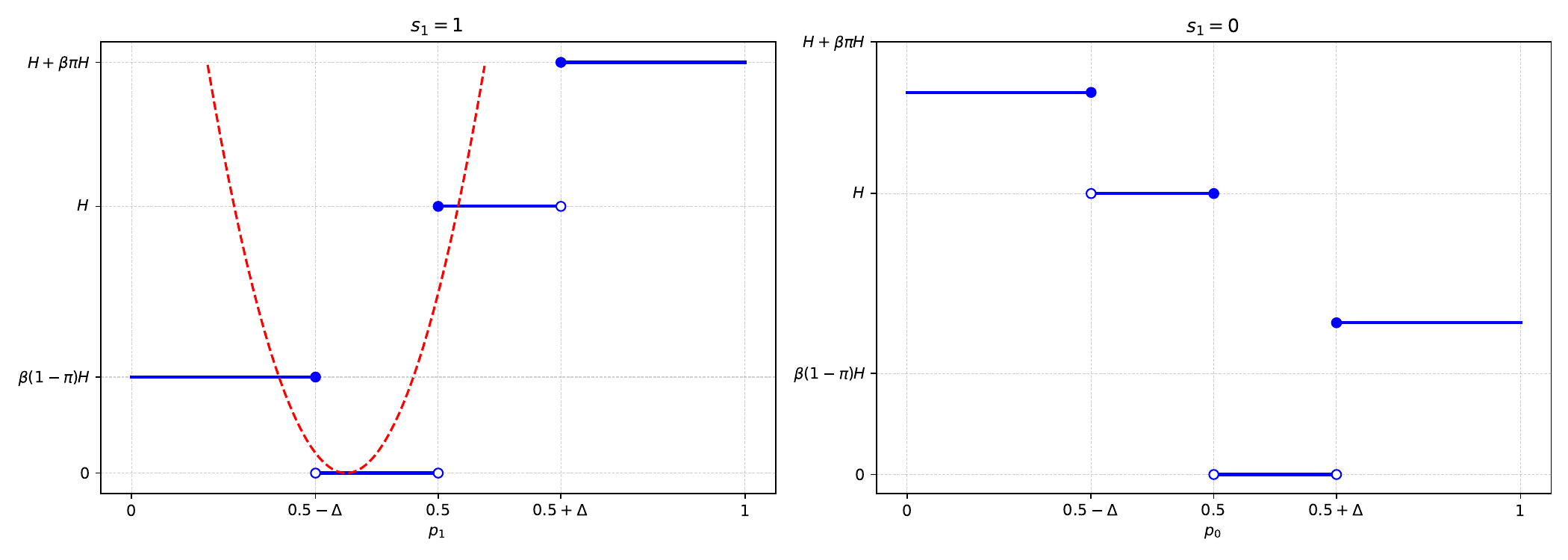}
    \caption{Benefits and Costs in the two elites Stackelberg Game}
    \label{fig:benefit_cost_2_elites}
\end{figure}

\subsection{Two competing Elites with a Long Horizon}

We now turn to the behavior of two elites facing a long horizon. We maintain the same assumptions, allowing $T \to \infty$. Next, we define the Bellman equation of each elite. Let $V_{A,s}(p)$ be $A$'s value when $A$ moves and observes $s$ at state $p$, and $U_A(p)$ be $A$'s value when $B$ moves from state $p$. Define $(V_{B,s},U_B)$ symmetrically. With the median rule, the Bellman system is
\begin{align}
V_{A,s}(p) &= \max_{p'\in[0,1]} \left\{
u_A\big(y(p'), s\big)
- c(p'-p) + \beta U_A(p')\right\} \label{eq:VA_clean}\\
U_A(p) &= \E_{s'\sim\pi}\Big[
u_A\big(y(\sigma_{B,s'}(p)), s'\big)
+\beta\, \E_{s''\sim\pi}V_{A,s''}(\sigma_{B,s'}(p))\Big], \label{eq:UA_clean}
\end{align}
and the analogous system for $(V_{B,s},U_B)$ with match/mismatch swapped. The Bellman equation decomposes each elite’s value into the immediate payoff, $u_A(y(p'), s) - c(p'-p)$, and the discounted continuation value, $\beta U_A(p')$. Here $V_{A,s}(p)$ is the value when $A$ moves and observes signal $s$, while $U_A(p)$ is the \emph{expected} value of elite $A$, when $B$ moves from state $p$, averaging over future states.

Following the literature on dynamic games with a single payoff-relevant aggregate state, we focus on Markov-Perfect Equilibria (MPE) in which strategies depend only on the current distribution of public opinion and the realization of signals. This rules out equilibria supported by non-Markovian, history-dependent threats and delivers a tractable yet still rich characterization.

Specifically, an MPE is a measurable profile $\{\sigma_{A,s},\sigma_{B,s}\}_{s\in\{0,1\}}$ and bounded value functions that solve \eqref{eq:VA_clean}–\eqref{eq:UA_clean} and their elite $B$-counterparts.

To understand how the different parameters shape equilibrium behavior between the two elites, we solve for Markov-perfect equilibrium (MPE) strategies numerically. In our environment, majority rule with a discontinuous threshold at \(p = \tfrac12\), strictly convex adjustment costs, and forward-looking strategic interaction between two elites generically produces non-concave, kinked value functions and multiple locally stable regions in the state space. This is exactly the type of setting for which the literature suggests using numerical MPE methods rather than closed-form characterization. For the numeric exercise, we assume a quadratic adjustment cost \(c(p' - p) = k(p' - p)^2\) with parameter \(k\), and set a long horizon \(T = 600\). We then compute the MPE by iterating on value and policy functions over a discretized state space, following standard practice in dynamic games and dynamic Industrial Organization.\footnote{See, for example, \citet{PakesMcGuire1994} and \citet{EricsonPakes1995} for classic algorithms in dynamic oligopoly, and \citet{DoraszelskiPakes2007} and \citet{AguirregabiriaMira2010} for surveys of numerical methods in dynamic games and dynamic discrete choice.}

Figure \ref{fig:2Elites_policyFunction} shows the resulting policy function for various parameter values. Relative to the single-elite case, the policy function no longer necessarily implies convergence to maximal polarization and, depending on the parameter values, can push toward higher or lower polarization. Therefore we again see that polarized elites can act as a damping force for polarization, introducing an incentive to lock in public opinion. Figure \ref{fig:2Elites_valueFunction} in the Appendix shows the corresponding value function. The figure shows that value functions are in general non-convex and not increasing. Moreover, compared to the single-elite case, the value is \emph{not} maximized at $p_1 = \frac{1}{2}$ and, depending on the parameter values, can exhibit a strong pull toward a more cohesive society with lower polarization. 

How does better technology affect behavior in the two-elite case? Panels (a)–(d) in Figure \ref{fig:2Elites_policyFunction} show the policy for different cost functions. When the cost is high (the green line), there are large regions of inaction and these regions are absorbing, capturing the same intuition as in the Two-Period game: preventing the opposition from enacting its desired policy can induce elites to maintain a cohesive society with low polarization. When the costs of changing public preferences become cheaper for both elites, there is a stronger pull toward high polarization, $p=\tfrac12$. When the cost is low, there is convergence toward maximal polarization. The reason is that there is no ability to deter the opposition from enacting its policy, and therefore it is better to polarize society to be able to enact the desired policy when possible. 

Panels (e)–(f) in Figure \ref{fig:2Elites_policyFunction} show the potential effect of lowering uncertainty due to better predictive technology or due to lower policy uncertainty in less uncertain times. The figure shows that when uncertainty is high, $\pi=\tfrac{1}{2}$, and cost is relatively high ($k=10$), elites are in more cases likely to push toward a cohesive society when current public opinion is aligned with their ideal goal, and toward a more polarized society in the reverse case. As uncertainty goes down (the orange and green lines), there is a stronger push for polarization in the less likely state, but a stronger push toward social cohesion in the more likely state, again capturing the fact that polarization among the elites can push the general masses toward more social cohesion. 

Finally, Figure \ref{fig:twoElitesBetaH_policyFunction} and \ref{fig:twoElitesBetaH_valueFunction} in the Appendix show how changes in the reward $H$ and the discount factor $\beta$ affect the policy function and value function. In general, as the reward increases there is a stronger push toward polarization. Looking at the discount factor, the figure shows that increasing $\beta$ makes elites more forward-looking, so they are willing to pay higher current costs to place public opinion in “good” regions for the future, thereby shrinking inaction regions and pulling public opinion more strongly toward high-value states—either consensus (low polarization) when they currently have the majority, or more polarized configurations when they are disadvantaged—thus amplifying the dynamic forces behind both polarization and lock-in.

\begin{figure}[H]
    \centering
    \includegraphics[width=\linewidth]{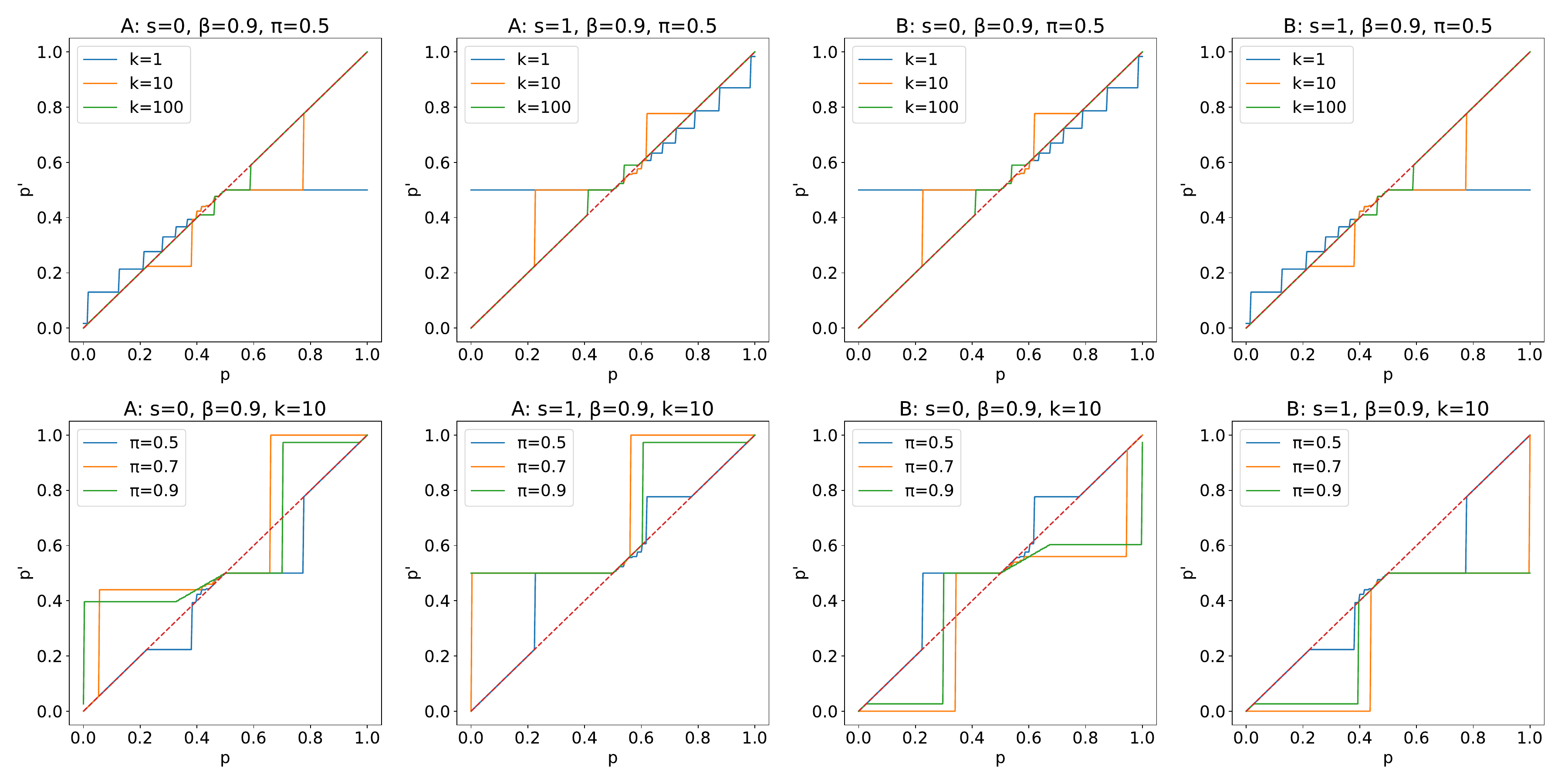}
    \caption{Two Competing Elites - Policy Function}
    \label{fig:2Elites_policyFunction}
\end{figure}

\section{Conclusions}\label{sec:conclusions}

Democratic governance rests on a fundamental tension: policies require public support to be implemented, yet shifting public opinion is costly. This paper introduces a novel mechanism that emerges from this tension—when elites must secure majority approval under uncertainty, they face strong incentives to manufacture polarization. The driving force is strategic insurance: a polarized society, with public opinion clustered near the decision threshold, minimizes the costs of aligning the public preferences with the future policy shocks. Counter-intuitively, the very institutions designed to constrain power majority rule and consensus requirements create conditions under which polarization becomes optimal from the elite's perspective.

Three results crystallize this logic. First, when there is a single elite it finds useful to pull society toward the highest polarization. A highly polarized society becomes an attractor because it lowers the cost needed in the future to adjust  expected future adjustment costs. Second: cheaper persuasion accelerates polarization. As AI slashes the cost of influence—through LLMs, microtargeting, algorithmic amplification—elites move faster toward divided equilibria. Every technical advance that makes persuasion easier tilts the calculus toward fragmentation. Third: competition among the elites can reverse the dynamic. When rivals alternate power, today's leader faces a trade-off: polarization offers flexibility but hands opponents a society easy to recapture. This generates "lock-in" incentives, where elites foster cohesion to freeze out future competitors. Rivalry can dampen or amplify polarization depending on persuasion costs, control persistence, and policy volatility.

These findings connect directly to the recent literature on transformative AI. Recent advances in large language models, automated content generation, and real-time behavioral targeting have dramatically reduced the technical barriers to large-scale persuasion. Where traditional influence campaigns required extensive infrastructure—broadcasting networks, coordinated messaging, costly ground operations—AI enables precision influence at unprecedented scale and speed. Our model suggests this technological shift will have first-order political consequences. As AI reduces the cost of shaping the masses preferences, our results predict accelerated movement toward polarized equilibria in societies with a single-elite regimes. The results are less certain in the case where there are two competing elites, where there could be push towards more or less polarization, depending on the other driving forces. 

Importantly, our framework abstracts from the documented costs of polarization- including erosion of social cohesion, increased risk of political violence, and reduced capacity for collective problem-solving. A complete welfare analysis would weigh the elites' adjustment-cost savings against these social externalities. When polarization imposes substantial negative spillovers not internalized by decision-makers, even benevolent elites may generate inefficient outcomes.

The rise of AI-driven persuasion technologies marks a critical juncture for democratic governance. Understanding how these tools interact with elite incentives—and the conditions under which their effects can be constrained—is essential for ensuring that advances in AI strengthen rather than destabilize democratic institutions. The question is not whether elites will adopt cheaper persuasion technologies, but whether our institutional architecture can channel their use toward socially beneficial ends.

\bibliographystyle{plainnat} 
\bibliography{bib}

\appendix
\section{Appendix}
\subsection{Figures}
\begin{figure}[H]
    \centering
    \includegraphics[width=\linewidth]{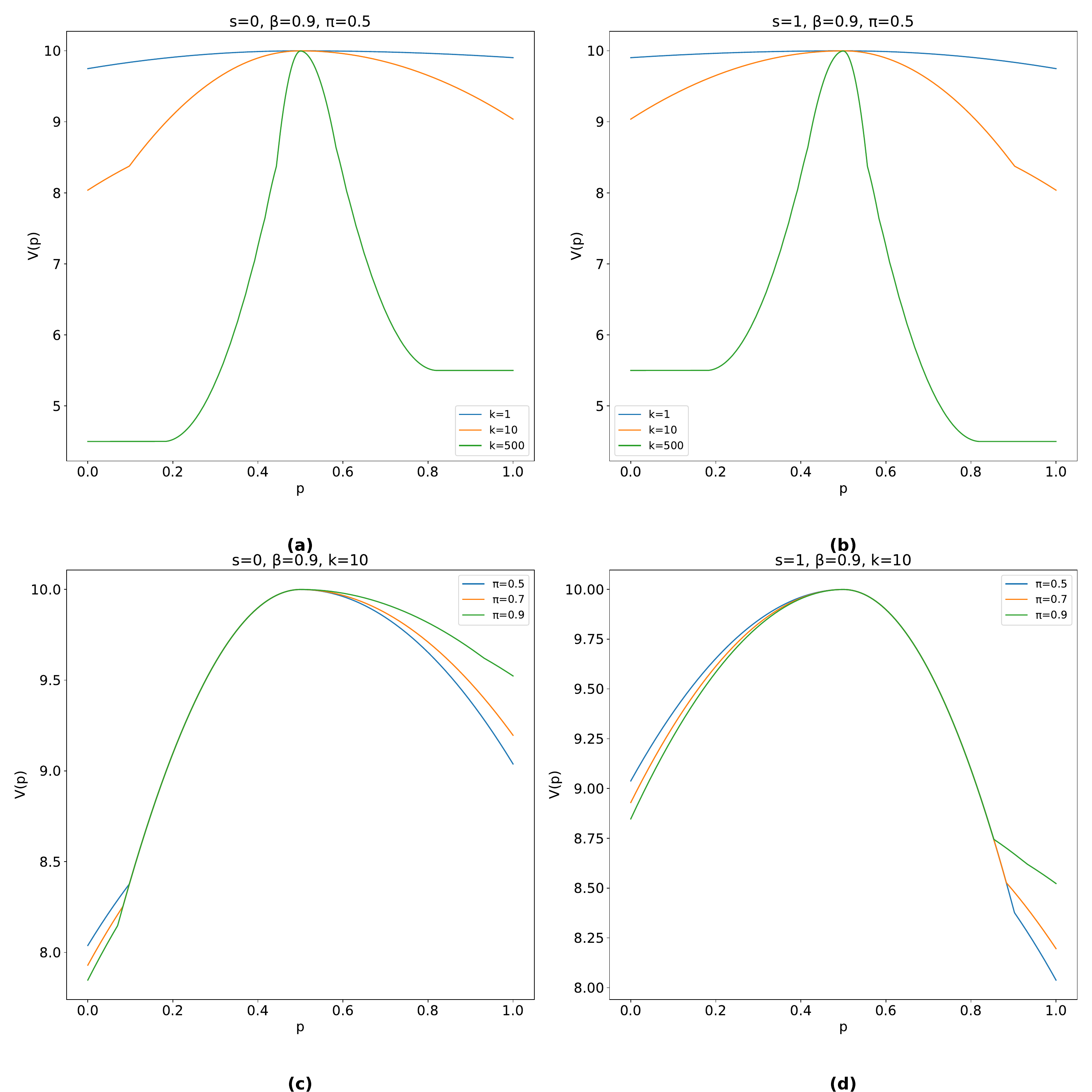}
    \caption{Single Elite - Value Function}
    \label{fig:oneEliteInfHorizon_valueFunction}
\end{figure}

\begin{figure}[H]
    \centering
    \includegraphics[width=\linewidth]{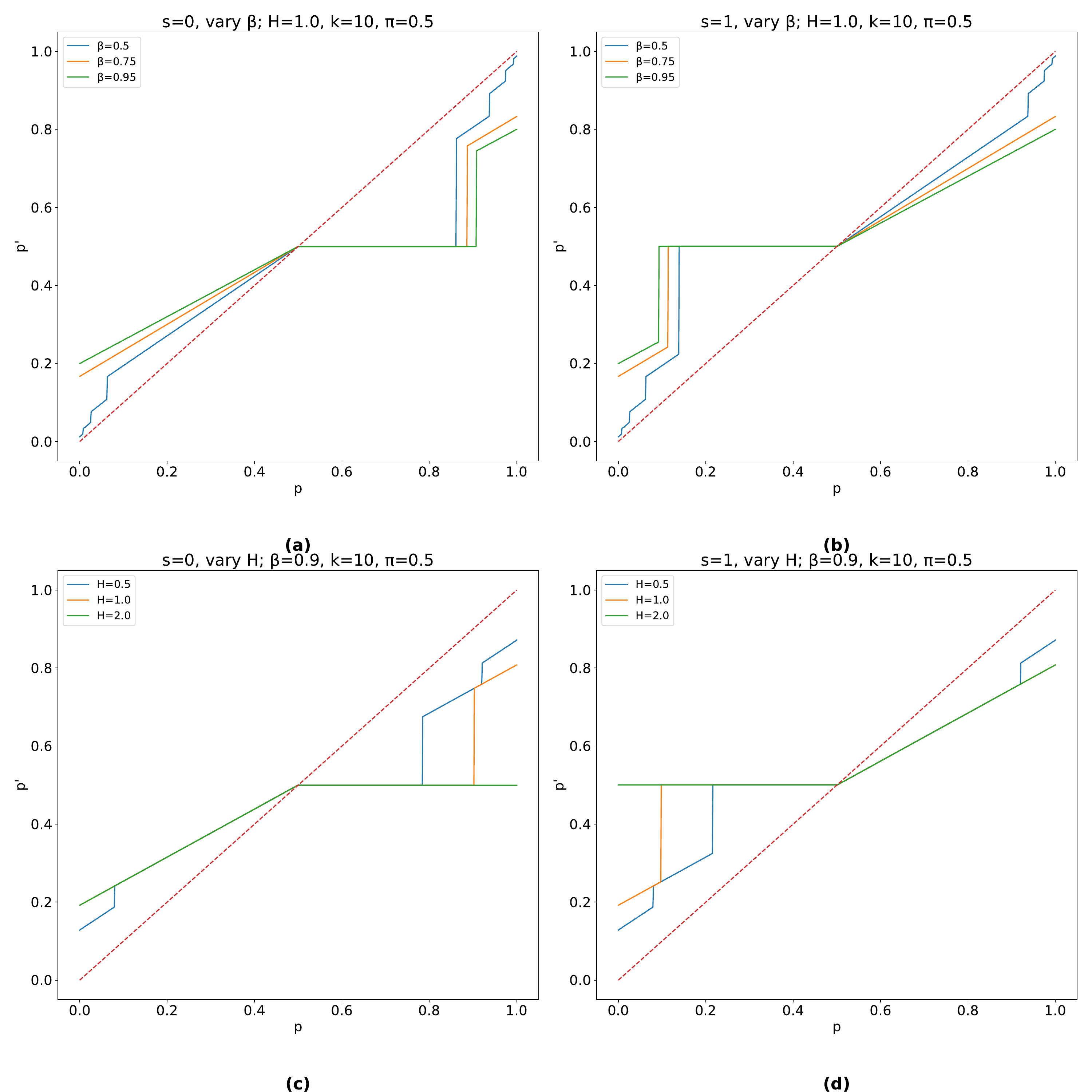}
    \caption{Single Elite - Policy Function - Variation in $\beta$ and $H$}
    \label{fig:singleEliteBetaH}
\end{figure}

\begin{figure}[H]
    \centering
    \includegraphics[width=\linewidth]{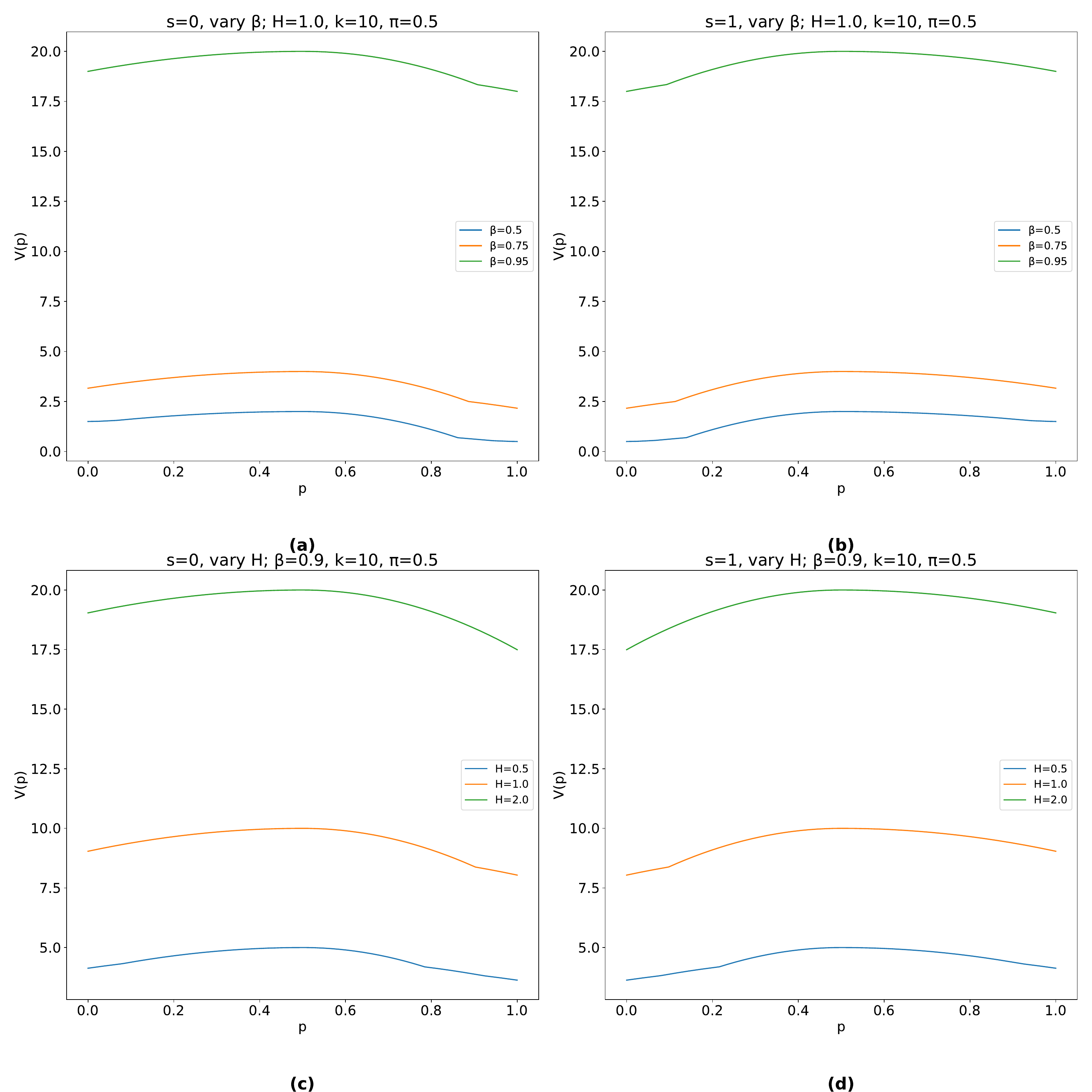}
    \caption{Single Elite - Value Function - Variation in $\beta$ and $H$}
    \label{fig:singleEliteBetaH_valueFunction}
\end{figure}

\begin{figure}[H]
    \centering
    \includegraphics[width=\linewidth]{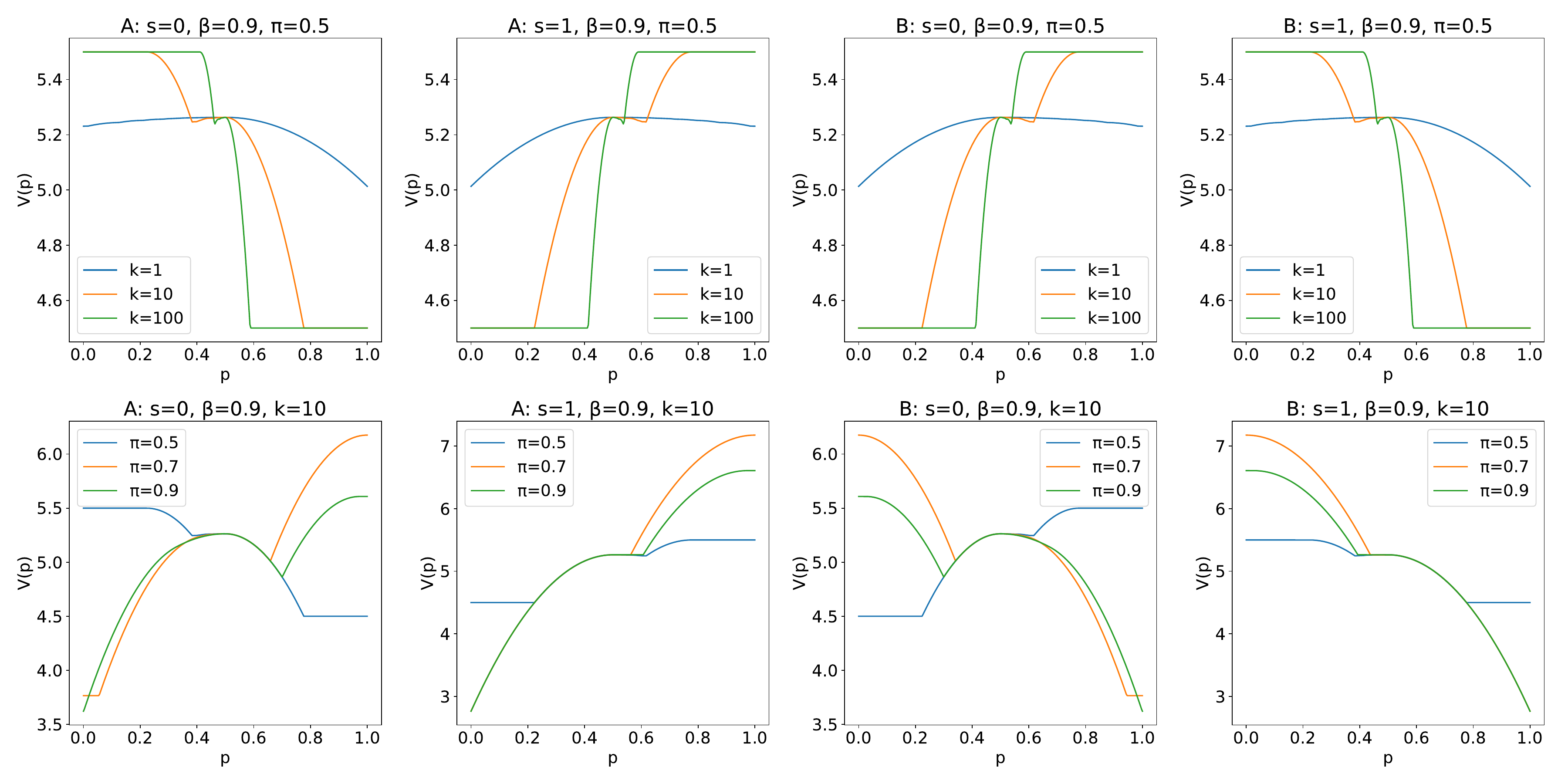}
    \caption{Two Competing Elites - Value Function}
    \label{fig:2Elites_valueFunction}
\end{figure}

\begin{figure}[H]
    \centering
    \includegraphics[width=\linewidth]{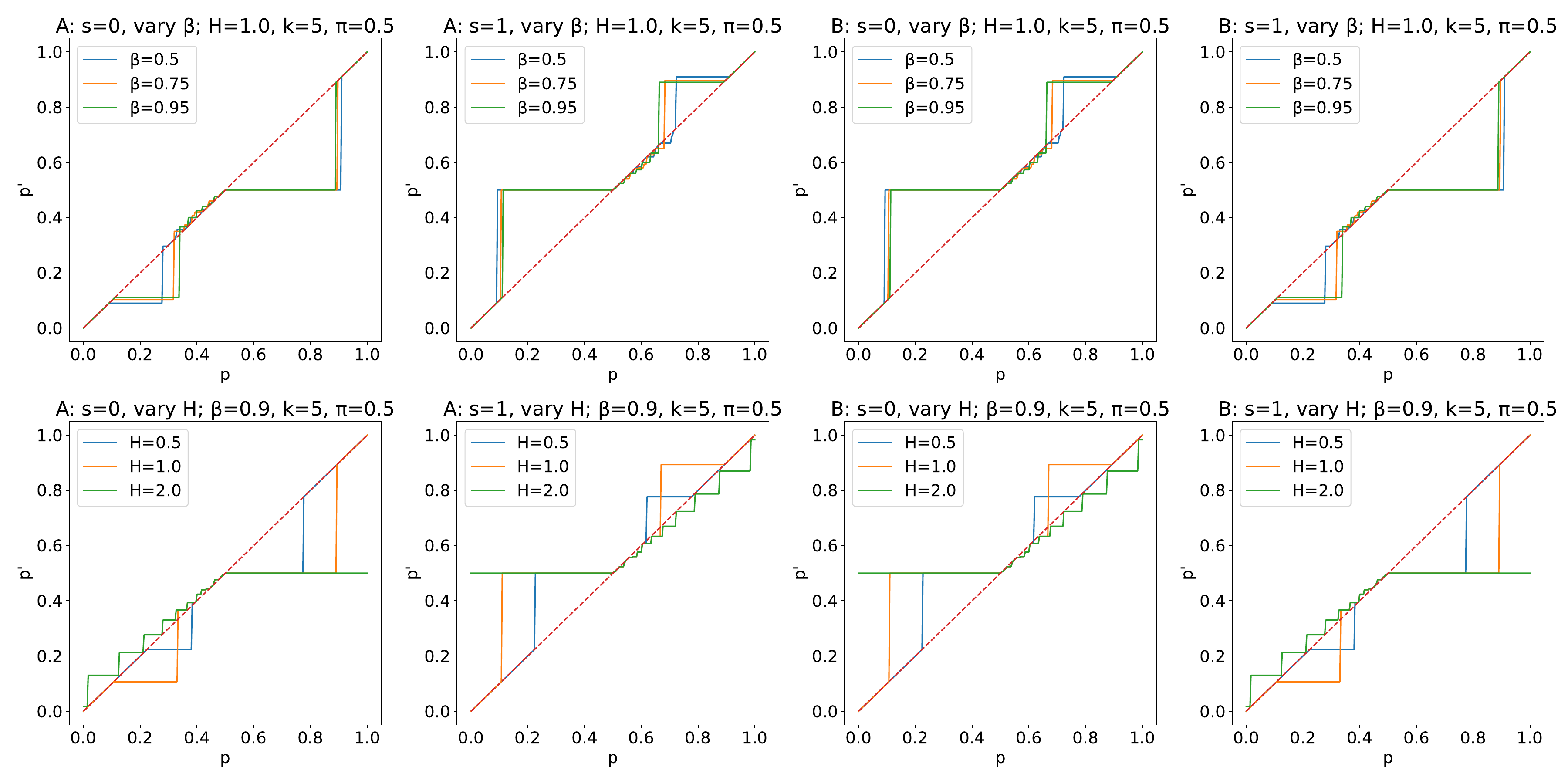}
    \caption{Two Competing Elites - Policy Function - Variation in $\beta$ and $H$}
    \label{fig:twoElitesBetaH_policyFunction}
\end{figure}

\begin{figure}[H]
    \centering
    \includegraphics[width=\linewidth]{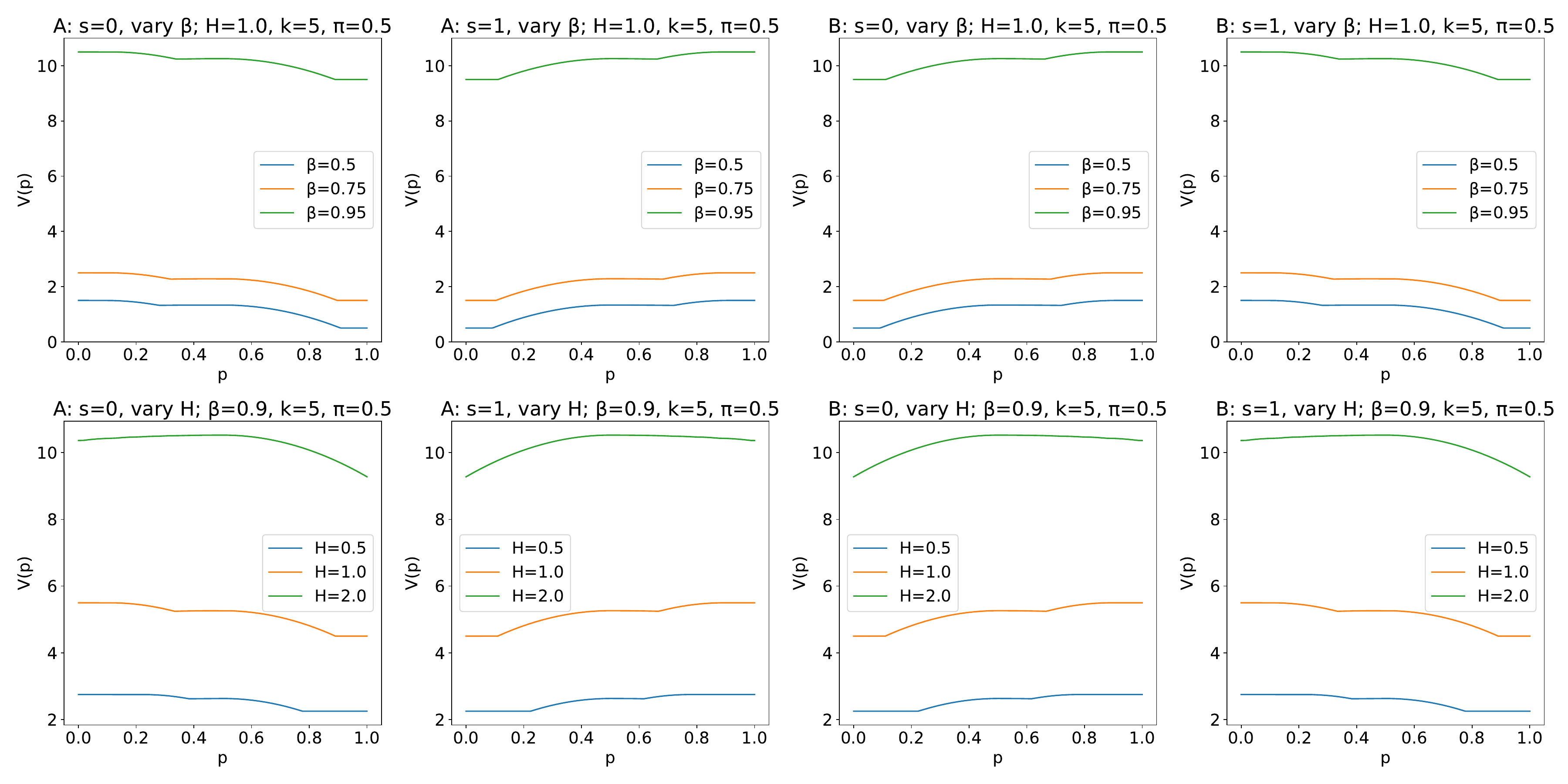}
    \caption{Two Competing Elites - Value Function - Variation in $\beta$ and $H$}
    \label{fig:twoElitesBetaH_valueFunction}
\end{figure}

\section{Proofs}
\subsection{Proof of Claim \ref{claim:oneElite_period2}}\label{app:proofClaim1}
\begin{proof}
Fix \(s=1\). If \(p\ge \tfrac12\) the elite already implements its preferred policy, so \(p'=p\) is optimal. If \(p<\tfrac12\), either move to \(\tfrac12\) to obtain \(H-c(\tfrac12-p)\) or do nothing and obtain \(0\); by the definition of \(p_0^*\), moving is optimal iff \(p>p_0^*\). Overshooting beyond \(\tfrac12\) is strictly dominated because the benefit is already \(H\) at \(\tfrac12\) while the cost strictly increases with distance. The case \(s=0\) is symmetric with \(p_1^*\). 
\end{proof}
\subsection{Proof of Claim \ref{claim:single_elite_two_period}}\label{app:proofClaim2}
\begin{proof}
Let
\[
J(p';p,s)\;:=\; H\,\mathbbm{1}\{y(p')=s\} - c(p'-p) + \beta\,\left[\pi V_{2,1}(p')+(1-\pi) V_{2,0}(p')\right],
\]
the period-1 objective for inherited support \(p\) and state \(s\).

\paragraph{Step 1 (Regional Maximiers).} 
We first consider the maximum $p'$ at each of the for region $A,B,C,D$. On \(A\) and \(D\), both the current payoff and the continuation term in \eqref{eq:EV2_piecewise} are constant over the region. Therefore \(J(\cdot;p,s)\) differs only by the cost \(c(p'-p)\), which is strictly increasing in \(|p'-p|\). Hence the regional maximizer on \(A\) (resp. \(D\)) is the  \(p'=p\) if \(p\in A\) (resp. \(p\in D\)), and otherwise the boundary \(p_0^*\) (resp. \(p_1^*\)).  

On \(B\), using the expected value in \eqref{eq:EV2_piecewise} and the fact that \(y(p')\) is constant on \(B\), maximizing \(J\) is equivalent to minimizing the strictly convex function, yielding the unique maximizer \(p_{B,\max}(p)\). The analysis on \(C\) is identical, giving \(p_{C,\max}(p)\).

\paragraph{Step 2 (Cutoffs \(p_0^*,p_1^*\) are not optimal).}
Suppose \(p>p_0^*\) and consider \(p'=p_0^*\). Moving a small amount \(\varepsilon>0\) to \(\tilde p:=p_0^*+\varepsilon\in B\) leaves the current payoff unchanged (\(y=0\) on both points), strictly reduces the cost (since \(|\tilde p-p|<|p_0^*-p|\)), and strictly increases the continuation value. Hence for every $p_0$, \(p_0^*\) is strictly dominated. We've seen that for each $p\leq p_0^*$ $p_0^*$ is dominated by $p$, and therefore $p_0^*$ cannot be a maximizer, unless the initial belief distribution $p_0$ is exactly $p_0^*$. The argument is symmetric for \(p_1^*\).

Taken together, the maximizer must be one of the regional maximizes \(\{p,\;p_{B,\max}(p),\;p_{C,\max}(p)\}\) or the boundary point \(p'=\tfrac12\), which can be strictly greater than nearby interior points because crossing the median changes the period-1 payoff by \(\pm H\) and, by tie-breaking, yields the full \(H\) at \(p'=\tfrac12\)). This proves part (i).

\paragraph{Step 3 (Polarization pull).}
Consider \(p\le\tfrac12\); the case \(p\ge\tfrac12\) is symmetric.  
If \(s=0\), any \(p'>\tfrac12\) strictly dominated by $p'=\frac{1}{2}$, as it lowers the current payoff (from \(H\) to \(0\)), increases today’s cost, and weakly lowers the continuation value. hence \(p'>\tfrac12\) cannot be optimal. Thus \(p'\le\tfrac12\).  
Similarly, if \(s=1\), any \(p' > \tfrac12\) is dominated by $p'=\frac{1}{2}$, as it equalizes the current payoff at \(H\), it has lower today’s cost and maximizes the continuation value within \(C\). Consequently, whenever crossing occurs, the chosen point is \(p'=\tfrac12\).  
Finally, on \(B\) the maximizer \(p_{B,\max}(p)\) lies in the interval \([p,\tfrac12]\): since \(c(q-p)\) and \(c(\tfrac12-q)\) are minimized at \(p\) and \(\tfrac12\) respectively, then their their strictly convex weighted sum, with positive weights, has a unique minimum between those two points. Therefore in all cases \(p'\in[p,\tfrac12]\), proving \(|p'-\tfrac12|\le |p-\tfrac12|\). This establishes (ii).
\end{proof}

\subsection*{Proof of Proposition \ref{claim:SingleInfiniteHorizon}}\label{app:proofProp1}

Before proving the proposition, we first introduce some definitions and notation.

\begin{definition}[Peak property]
A function \(f:[a,b]\to\mathbb{R}\) has the peak property at \(\alpha\in[a,b]\) if it is weakly increasing on \([a,\alpha]\) and weakly decreasing on \([\alpha,b]\).
\end{definition}

\paragraph{Function class and Bellman operator.}
Let $\mathcal{B}$ be the space of bounded functions $v:\{0,1\}\times[0,1]\to\mathbb{R}$ with sup norm
$\|v\|_\infty=\sup_{s,p}|v_s(p)|$.
Let
\[
\mathcal{F}:=\Big\{v\in\mathcal{B}:\ \forall s,\ v_s(\cdot)\text{ has the peak property at } \tfrac12\Big\}.
\]
For $v\in\mathcal{B}$ define the Bellman operator $T:\mathcal{B}\to\mathcal{B}$ by
\[
(Tv)_s(p):=\max_{p'\in[0,1]}\Big\{R(s,p')-c(p'-p)+\beta\,[\pi\,v_1(p')+(1-\pi)\,v_0(p')]\Big\}.
\]
Also define, for later use the returns function. This function assures assures that the value at $\frac{1}{2}$ is $H$, and correspond to $u(y(p'),s)$
\[
R(s,p') = \begin{cases}
    H\mathbbm{1}\{y(p') = s\} & \text{ if } p' \neq \frac{1}{2} \\
    H & \text{ if } p' = \tfrac{1}{2}
\end{cases}
\]
and the value of policy $v$
\[
\Psi_v(p';s,x):=R(s,p')-c(p'-x)+\beta\big[\pi\,v_1(p')+(1-\pi)\,v_0(p')\big],
\quad\text{so that }(Tv)_s(x)=\max_{p'}\Psi_v(p';s,x).
\]

To prove the first part of Proposition \ref{claim:SingleInfiniteHorizon}, we demonstrate that the value function \( V_s(p) \) satisfies the peak property at \( p = \frac{1}{2} \) over the interval \( [0,1] \). Our proof follows the standard argument presented in \cite{stokey1989recursive} to establish properties of the value function.

\paragraph{Step 1: Contraction (Blackwell).}
Using theorem 3.3 in \cite{stokey1989recursive} we show that T is a contraction. First, $T$ is monotone: $v\le w\Rightarrow Tv\le Tw$.  
$T$ is discounting: for any constant $a$, $T(v+a)=Tv+\beta a$.  
Hence $T$ is a $\beta$‑contraction on $(\mathcal{B},\|\cdot\|_\infty)$; a unique bounded fixed point $(V_0,V_1)$ exists.

\paragraph{}Next we prove a simple useful lemma that shows that no elite overshoot over $\frac{1}{2}$ or go backwards, if $v \in \mathcal{F}$. 

\begin{lemma}[Domain reduction: no overshooting]\label{lem:domain-reduction}
Fix $v\in\mathcal{F}$ and $p\le \tfrac12$. In $\max_{p'}\Psi_v(p';s,p)$ it is w.l.o.g. to restrict $p'\in[p,\tfrac12]$. (Symmetrically, if $p\ge\tfrac12$, restrict to $p'\in[\tfrac12,p]$.)
\end{lemma}

\begin{proof}
For $p\le\tfrac12$ and any $p'>\tfrac12$, replacing $p'$ by $\tilde p=\tfrac12$ weakly increases $R$, weakly reduces cost (shorter distance), and weakly increases the continuation term since each $v_s(\cdot)$ peaks at $\tfrac12$. Any $p'<p$ is weakly dominated by $p$ by the same logic.
\end{proof}

\paragraph{Step 2: $T$ preserves the peak property.}
Let $v\in\mathcal{F}$ and fix $s$. We show $(Tv)_s(\cdot)$ is weakly increasing on $[0,\tfrac12]$; the other side is symmetric.

Fix $0\le p<\bar p\le\tfrac12$. By Lemma~\ref{lem:domain-reduction}, restrict to $p'\in[p,\tfrac12]$. Let $p^*\in\arg\max_{p'\in[p,\tfrac12]}\Psi_v(p';s,p)$, and set
\[
\bar p''\ :=\ \min\!\big\{\ \tfrac12,\ \bar p+(p^*-p)\ \big\}.
\]
Note $\bar p\le \bar p''\le \tfrac12$ and $\bar p''\ge p^*$. Then
\[
\begin{aligned}
(Tv)_s(p)
&= \Psi_v(p^*;s,p) \\
&\le R(s,p^*) - c(p^*-p) + \beta\big[\pi\,v_1(\bar p'')+(1-\pi)\,v_0(\bar p'')\big] 
&&\\
&\le R(s,\bar p'') - c(\bar p''-\bar p) + \beta\big[\pi\,v_1(\bar p'')+(1-\pi)\,v_0(\bar p'')\big]
&&\\
&\le \max_{q\in[\bar p,\tfrac12]}\Psi_v(q;s,\bar p)
\ =\ (Tv)_s(\bar p),
\end{aligned}
\]
where the first inequality is due to  $\frac{1}{2} \ge \bar p''\ge p^*$ and $v_s$ nondecreasing on $[0,\tfrac12]$), the second inequality follows since $R$ nondecreasing on $[0,\tfrac12]$, and  $\bar p''-\bar p\le p^*-p$. Thus $(Tv)_s$ is weakly increasing on $[0,\tfrac12]$.

\paragraph{Step 3: $\mathcal{F}$ is complete (closed under uniform limits).}
$(\mathcal{B},\|\cdot\|_\infty)$ is a Banach space. Let  $f_n\in\mathcal{F}$ be a cauchy sequence and therefore we have $f_n\to f \in \mathcal{B}$ uniformly, then for $p<\bar p\le\tfrac12$, $f_n(s,p)\le f_n(s,\bar p)$ for all $n$ implies $f(s,p)\le f(s,\bar p)$. The argument on $[\tfrac12,1]$ is symmetric. Hence $\mathcal{F}$ is closed in $\mathcal{B}$ and therefore complete.

\medskip
\medskip
\noindent\emph{Conclusion after Steps 2–3.}
Thus $T(\mathcal{F})\subseteq\mathcal{F}$. Since $\mathcal{F}$ is closed under uniform limits (Step 3) and $T$ is a contraction (Step 1), the unique fixed point $V$ also lies in $\mathcal{F}$, proving Proposition~\ref{claim:SingleInfiniteHorizon}(1).

\paragraph{Step 4: Polarization pull.}
By Lemma~\ref{lem:domain-reduction}, any maximizer for $(TV)_s(p)$ lies in $[p,\tfrac12]$ if $p\le\tfrac12$ and in $[\tfrac12,p]$ if $p\ge\tfrac12$. Hence
\[
\min\{p,\tfrac12\}\ \le\ \sigma(s,p)\ \le\ \max\{p,\tfrac12\},
\quad\text{so}\quad |\sigma(s,p)-\tfrac12|\ \le\ |p-\tfrac12|.
\]
At $p=\tfrac12$ any move raises cost and (weakly) lowers value, so $\sigma(s,\tfrac12)=\tfrac12$.

\paragraph{Step 5: Monotonicity of the optimal policy in $p$.}
Let $\Psi:=\Psi_V$ and fix $s$. Suppose $0\le p<\bar p\le \tfrac12$.
Let $p^*\in\arg\max_{q\in[p,\tfrac12]}\Psi(q;s,p)$ and
$\bar p^*\in\arg\max_{q\in[\bar p,\tfrac12]}\Psi(q;s,\bar p)$.
Define
\[
\Delta(q):=\Psi(q;s,p)-\Psi(q;s,\bar p)=-c(q-p)+c(q-\bar p).
\]

We claim that $\Delta(\cdot)$ is strictly decreasing on $[0,1]$. Using strict convexity of $c$ on $[0,\infty)$, the increment
$x\mapsto c(x+a)-c(x)$ is strictly increasing in $x$. Let $q_2>q_1$ and define set $a=q_2-q_1>0$, $x_1 = q_1 - \bar p$ and $x_2 = q_1 - p $. We have $x_2 > x_1$ and so
\[
\big[c(q_2-p)-c(q_1-p)\big]\ >\ \big[c(q_2-\bar p)-c(q_1-\bar p)\big].
\]
Rearranging,
\[
\Delta(q_2)-\Delta(q_1)\ =\ -\big[c(q_2-p)-c(q_1-p)\big]+\big[c(q_2-\bar p)-c(q_1-\bar p)\big]\ <\ 0.
\]
Therefore $\Delta$ is strictly decreasing on $[0,1]$.

finally, assume towards contradiction that $\bar p^*<p^*$. Then, from optimality we have:
\[
\Psi(p^*;s,p)\ge \Psi(\bar p^*;s,p),\qquad
\Psi(\bar p^*;s,\bar p)\ge \Psi(p^*;s,\bar p),
\]
adding the two inequalities implies  $\Delta(p^*)\ge \Delta(\bar p^*)$, contradicting the strict decrease together with $\bar p^*<p^*$. Therefore $p^*\le \bar p^*$; the case $\tfrac12\le \bar p<p\le 1$ is symmetric. This proves the monotonicity statements in Proposition~\ref{claim:SingleInfiniteHorizon}(2).
\qed

\subsection{Proof of claim \ref{lem:shrink-move}}\label{app:proofClaim3}

\begin{proof}
We do the case $p\le\tfrac12$; the other case is symmetric. By domain reduction (Lemma~\ref{lem:domain-reduction}), $p^*_c,p^*_{\tilde c}\in[p,\tfrac12]$.
Suppose, to the contrary, that $p^*_{\tilde c}>p^*_c$.
By optimality,
\[
\Psi_c(p^*_c;s,p)\ \ge\ \Psi_c(p^*_{\tilde c};s,p),\qquad
\Psi_{\tilde c}(p^*_{\tilde c};s,p)\ \ge\ \Psi_{\tilde c}(p^*_c;s,p).
\]
Adding and rearranging gives
\[
\big[\tilde c(p^*_{\tilde c}-p)-\tilde c(p^*_c-p)\big]\ \le\ \big[c(p^*_{\tilde c}-p)-c(p^*_c-p)\big].
\]
Let $x_1:=p^*_c-p$ and $x_2:=p^*_{\tilde c}-p$; then $0\le x_1<x_2$. Since $\tilde c$ cost-dominates $c$,
\[
\tilde c(x_2)-\tilde c(x_1)\ >\ c(x_2)-c(x_1),
\]
a contradiction. Hence $p^*_{\tilde c}\le p^*_c$.
\end{proof}

\subsection{Proof of claim \ref{claim:2p_stackleberg}}\label{app:proofClaim4}
\begin{proof}

We solve by backward induction and start with period 2. Take the public opinion at the beginning of period 2 to be \(p_1\), and let \(1-s_2 \in {0,1}\) be Elite \(B\)’s preferred side in period 2. Recall that the implemented policy in period 2 is
\[
y(p_1) = \mathbbm{1}\{p_1 \ge \tfrac12\}.
\]
Elite \(B\) gets benefit \(H\) if the implemented policy matches her preferred side, i.e. if \(y(p_1)=1-s_2
\); otherwise they get 0. Moving public opinion from \(p_1\) to some other point costs \(c(|p - p_1|)\).

There are two situations. First, the current policy already matches \(B\)’s preference. This is the case when $y(p_1) = 1- s_2$. Then Elite \(B\) already gets (H) with zero cost. Any move away from \(p_1\) strictly reduces her payoff, because it adds a positive cost and does not increase the benefit. Hence the optimal action is to do nothing and set \(p = p_1\).

Second, if the current policy is the opposite of \(B\)’s preference. This is the case when $y(p_1) = s_2$. Then Elite \(B\) currently gets 0. To get benefit \(H\), she must change the implemented policy. The cheapest way to change the implemented policy is to move public opinion exactly to \(\tfrac12\): any move beyond \(\tfrac12\) yields the same policy but costs more. So the only candidate move is \(p = \tfrac12\). If the cost of that move is small enough, i.e. $c\big(\tfrac12 - p_1\big) \le H$\footnote{Notice that we assume that $c$ is symmetric, i.e. $c(\frac{1}{2}-p_1) = c(p_1-\frac{1}{2})$, therefore we can consider move from any direction to $\frac{1}{2}$} then moving to \(\tfrac12\) is optimal: Elite \(B\) pays at most \(H\) and gets benefit \(H\), and no farther move can do better because it gives the same benefit \(H\) but is more expensive. If instead the cost of moving to \(\tfrac12\) is \textit{higher} than \(H\), then even the cheapest policy-changing move is not worth it. Any move farther than \(\tfrac12\) would cost even more and still give at most \(H\). Hence in this case the optimal action is again to do nothing and keep \(p = p_1\). Putting these cases together gives exactly the strategy stated in the proposition.

Next, we move to show that the optimal strategy of elite $A$, given the strategy of elite $B$, is on one of four points $\{p_0,\tfrac{1}{2},\tfrac{1}{2}\pm \Delta\}$. Let $\Phi(p_1)$ be the next period value of elite $A$:

\[
\Phi(p_1) = \E_{s_2}\big[u_A(y(\sigma_{B,s_2}(p_1)),s_2)\big]=
\begin{cases}
\pi H, & p_1\ge \tfrac12+\Delta,\\
(1-\pi)H, & p_1\le \tfrac12-\Delta,\\
0, & p_1=\tfrac12 \text{ or } |p_1-\tfrac12|<\Delta.
\end{cases}
\]

We consider three cases.

\begin{enumerate}
    \item \textbf{Case 1: \(p_0 \le \tfrac12 - \Delta\)}. Consider any feasible \(p_1\).
    \begin{itemize}
        \item \(p_1 < \tfrac12-\Delta\): then \(\Phi(p_1)=(1-\pi)H\). Hence the value is given by 
       \[
       u_A(0,s_1) - c(p_1 - p_0) + \beta (1-\pi)H.
       \]
       Since \(c(\cdot)\) is strictly increasing then choosing \(p_1 \neq p_0\), gives strictly smaller value choosing \(p_1=p_0\). Thus every \(p_1 < \tfrac12-\Delta\) is dominated by \(p_1=p_0\).
        \item \(p_1 = \tfrac12-\Delta
   \): this point yields
   \[
   u_A(0,s_1) - c\big(p_0 - (\tfrac12-\Delta)\big) + \beta (1-\pi)H,
   \]
   which is exactly the “semi-lock left” value in the proposition.
    \item \(p_1 \in (\tfrac12-\Delta, \tfrac12+\Delta)\) and  \(p_1 \neq \tfrac12\): for such \(p_1\), we have \(\Phi(p_1)=0\). Then
   \[
   u_A(y(p_1),s_1) - c(p_1 - p_0) + \beta\cdot 0
   \]
   but in all these cases Elite \(A\) forgoes the future value that it gets at \(p_0\). Hence every such point is \(weakly\) dominated either by staying at \(p_0\) or by moving to \(\tfrac12\), as seen next.
    \item \(p_1 = \tfrac12\): this point yields
   \[
   H - c\big(p_0-\tfrac12\big) + \beta\cdot 0,
   \]
    \item Any \(p_1 \ge \tfrac12+\Delta\): then \(\Phi(p_1)=\pi H\), and the value is
   \[
   u_A(1,s_1) - c(p_1 - p_0) + \beta \pi H.
   \]
   Among all such \(p_1\), the closest to \(p_0\) is \(p_1 = \tfrac12+\Delta\). Every \(p_1 > \tfrac12+\Delta\) gives exactly the same current policy and the same future value but higher cost, hence is dominated by \(\tfrac12+\Delta\). So the only non-dominated right-point in this case is
   \[
   p_1 = \tfrac12+\Delta \quad \Rightarrow \quad u_A(1,s_1) - c\big((\tfrac12+\Delta)-p_0\big) + \beta \pi H.
   \]
    \end{itemize}
    \item  \textbf{Case 2: \(p_0 \in (\tfrac12-\Delta, \tfrac12+\Delta)\)}
    \begin{itemize}
        \item \(p_1=p_0\):This gives
   \[
   u_A\big(y(p_0), s_1\big) + \beta \cdot 0,
   \]
   which is the “Inaction” value.
   \item Move to another point in \((\tfrac12-\Delta, \tfrac12+\Delta)\), \(p_1 \neq \tfrac12\):
   For any such point \(\Phi(p_1)=0\). If \(p_1\) is on the same side of \(\tfrac12\) as \(p_0\), then \(y(p_1)=y(p_0)\), so the only difference is the cost, which is higher at \(p_1\neq p_0\). Thus dominated by \(p_1=p_0\). If \(p_1\) is on the other side of \(\tfrac12\), then moving exactly to \(\tfrac12\) is closer to \(p_0\) than moving to any interior point on the other side, and also delivers the “Median” current payoff. Thus such interior crossing points are dominated by \(p_1=\tfrac12\). Hence all interior-band \(p_1 \neq \tfrac12\) are dominated.
    \item Move to the median \(p_1=\tfrac12\):
   This yields
   \[
   H - c\big(p_0 - \tfrac12\big) + \beta\cdot 0.
   \]
    \item  Move to the right boundary \(p_1 = \tfrac12+\Delta\): This yields the right-side protected value \(\beta \pi H\). It gives
       \[
       u_A(1,s_1) - c\big((\tfrac12+\Delta)-p_0\big) + \beta \pi H.
       \]
   \item Move to any \(p_1 > \tfrac12+\Delta\). This gives the same \(\beta \pi H\), but is more expensive; hence such points are dominated by \(\tfrac12+\Delta\).
    \item Move to the left boundary \(p_1 = \tfrac12-\Delta\) gives 
   \[
   u_A(0,s_1) - c\big(p_0-(\tfrac12-\Delta)\big) + \beta (1-\pi)H,
   \]
   \item \(p_1 < \tfrac12-\Delta\) is more expensive than moving to $p_1 = \tfrac12-\Delta$ and hence dominated.
    \end{itemize}
    \item \textbf{Case 3: \(p_0 \ge \tfrac12 + \Delta\)} This case is symmetric to Case 1.
\end{enumerate}

Since for every possible initial public opinion \(p_0\) the set of non-dominated period-1 choices for Elite (A) is contained in $
{p_0, \tfrac12, \tfrac12+\Delta, \tfrac12-\Delta}$
and since the elite A's objective evaluated at these four points is exactly the four scalars listed in the statement, it follows that in equilibrium Elite \(A\) choose among these four values and picks the largest one. This concludes the proof. 
\end{proof}

\end{document}